\documentclass[12pt]{article}
\usepackage{amsmath}
\usepackage{graphicx,psfrag}
\usepackage{epsf}
\usepackage{enumerate}
\usepackage{float} 
\usepackage{multirow}
\usepackage{multicol}
\RequirePackage{amsthm,amsmath}
\usepackage{graphicx} 
\usepackage{booktabs} 
\usepackage{caption}
\usepackage{subcaption}
\usepackage{mathtools}
\usepackage{array}
\usepackage{calc}
\usepackage{xcolor}
\usepackage{hyperref}
\usepackage{amsfonts}
\usepackage{setspace}

\usepackage{natbib}
\newcommand{\blind}{0}

\newtheorem{remark}{Remark}

\newtheorem{theorem}{Theorem}
\newtheorem{lemma}{Lemma}

\begin{document}

\def\spacingset#1{\renewcommand{\baselinestretch}%
{#1}\small\normalsize} \spacingset{1}


\if0\blind
{
  \title{\bf Inference for Joint Quantile and Expected Shortfall Regression}
  \author{Xiang Peng
  \hspace{.2cm}\\
    Department of Statistics, George Washington University\\
    and \\
    Huixia Judy Wang \\
    Department of Statistics, George Washington University}
  \maketitle
} \fi

\if1\blind
{
  \bigskip
  \bigskip
  \bigskip
  \begin{center}
    {\LARGE\bf Title}
\end{center}
  \medskip
} \fi

\bigskip
\begin{abstract}
	Quantiles and expected shortfalls are commonly used risk measures in financial risk management. The two measurements are correlated while have distinguished features. In this project, our primary goal is to develop stable and practical inference method for conditional expected shortfall. To facilitate the statistical inference procedure, we consider the joint modeling of conditional quantile and expected shortfall. While the regression coefficients can be estimated jointly by minimizing a class of strictly consistent joint loss functions, the computation is challenging especially when the dimension of parameters is large since the loss functions are neither differentiable nor convex. To reduce the computational effort, we propose a two-step estimation procedure by first estimating the quantile regression parameters with standard quantile regression. We show that the two-step estimator has the same asymptotic properties as the joint estimator, but the former is numerically more efficient. We further develop a score-type inference method for hypothesis testing and confidence interval construction. Compared to the Wald-type method, the score method is robust against heterogeneity and is superior in finite samples, especially for cases with a large number of confounding factors. We demonstrate the advantages of the proposed methods over existing approaches through numerical studies.
\end{abstract}

\noindent%

\spacingset{1.45}
\title{}
\section{Introduction}\label{sec:intro_ES}
A tail quantile of the profit-and-loss distribution measures the risk of loss for investments, which is known as Value-at-Risk (VaR) and has
been widely applied for capital allocation and risk management over the past two decades \citep{mcneil2015quantitative}. Although it is an intuitive measure, VaR has been criticized since it fails to capture tail risks beyond itself. Expected Shortfall (ES), defined as the average above or below a certain quantile, fulfills such deficiency as it better characterizes the tail behavior by consolidating information from the entire tail region. In addition, ES has the desired property of subadditivity which VaR lacks in general \citep{artzner1997thinking, artzner1999coherent}.
With these appealing features, ES has attracted increasing attentions and has been more widely applied for risk management in recent years. In financial risk management, the Basel Committee recently to shift the quantitative risk metrics system from VaR to ES \citep{committee_2013}.

In many applications, risk measures might depend on exogenous covariates. For instance, market risks are often change across investment conditions, such as macroeconomic, financial, and political environments. For most clinical studies, patient outcomes are usually associated with demographic and therapeutic information. It is thus of interest to focus on the conditional risk measures adjusting for certain covariates. In this project, we consider the problem of inference for conditional expected shortfall (CES).

\cite{he2010detection} introduced a COVariate-adjusted Expected Shortfall (\textit{COVES}) test to detect treatment effects through CES, which is motivated by a clinical study with balanced design. However, as shown in Sections \ref{sec:simulation_ES} and \ref{sec:real_data_ES},
the statistical power of the \textit{COVES} test may be affected when there are unbalanced covariates. 
To evaluate CES beyond the scope of treatment differences, we consider inference based on a regression framework. Nevertheless, as pointed out by \cite{gneiting2011making}, CES is not ``elicitable" in the sense that it cannot be represented as the minimizer of an expected loss, and hence the stand-alone regression for CES is infeasible.
To overcome the problem of ``elicitability," \cite{leorato2012asymptotically} and \cite{peracchi2008estimating} suggested to approximate CES by fitting an entire quantile process, which imposes both computational and theoretical challenges. Alternative methods such as those proposed by \cite{cai2008nonparametric}, \cite{kato2012weighted} and \cite{xiao2014right} rely on kernel-smoothing estimation for the conditional distribution function, which are subject to the ``curse-of-dimensionality" and practically feasible only for data with a few covariates. 

For inference on CES, we consider an alternative approach through the joint modeling of conditional quantile and CES. \cite{fissler2016higher} recently showed that VaR and ES are jointly ``elicitable," and they provided a class of strictly consistent joint loss functions for the pairs of quantile and ES at the same probability level. \cite{dimitriadis2019joint} utilized the joint loss functions in a regression setup for quantile and ES. The computation of the joint estimator is challenging especially when the dimension of parameters is large since the joint loss function is neither differentiable nor convex. To reduce the computational effort, we propose a two-step estimation procedure. We first estimate the quantile parameters with standard quantile regression \citep{koenkerquantile}, and then estimate the ES regression coefficients bt minimizing the simplified objective function with the quantile estimators plugged in. We show that the two-step estimator has the same asymptotic properties as the joint estimator, but the former is numerically more efficient. In addition, the CES estimation in the second step is locally robust to the quantile estimation in the first step, which implies that the local misspecification of the quantile parameters has no effect on the asymptotic distribution of the ES estimator; see \cite{chernozhukov2016locally} for an elaboration on local robustness.

The Wald-type inference method can be conducted based on the asymptotic distribution of the parameter estimator.
However, it has been shown in quantile regression literature that the Wald-type test is generally unstable for small sample sizes, partly due to the uncertainty from estimating nuisance parameters involved in the asymptotic variance, such as the conditional densities of the response \citep{chen2005computational, kocherginsky2005practical}. We develop a score-type inference method for hypothesis testing and confidence interval construction. Numerical studies suggest that the proposed score-type method is superior to the Wald-type method in finite samples, especially when the data is heterogeneous and involves a large number of confounding factors. Furthermore, the method 
provides more accurate results than the \textit{COVES} approach for unbalanced design. 

In Section \ref{sec:proposed_method_ES}, we first present the two-step estimation procedure for the joint regression model and the large sample properties of the resulting estimators, and then develop the score-type inference method for the ES regression parameters. We assess the finite sample performance of the proposed inference procedure with simulation studies in Section \ref{sec:simulation_ES}. The merit of the proposed method is illustrated by analyzing two real data sets in Section \ref{sec:real_data_ES}. Some concluding remarks are provided in Section \ref{sec:conclusion_ES}. Proofs are deferred to the Appendix.

\section{Proposed Method}\label{sec:proposed_method_ES}

\subsection{Joint regression model}\label{subsec:joint_reg_model_ES}

Consider a continuous response $Y$ and a $p$-dimensional design vector $\mathbf{X}$. At a given probability level $\tau \in (0, 1)$, the conditional quantile of $Y$ given $\mathbf{X}$ is defined as
$$
Q_\tau(Y | \mathbf{X}) = \inf\{y \in \mathbb{R}: F^{-1}_{Y|\mathbf{X}}(y) \ge \tau\},
$$ 
where $F_{Y|\mathbf{X}}$ is the conditional distribution function of $Y$ given $\mathbf{X}$. The corresponding CES is defined as $$
ES_\tau(Y|\mathbf{X}) = \tau^{-1} \int_0^\tau F^{-1}_{Y|\mathbf{X}} (u) du,
$$ 
which is deemed to be more informative than the conditional quantile as CES summarizes the entire tail behavior of the conditional distribution.
In this project, we are interested in the inference for the CES of $Y$ given $\mathbf{X}$ at a certain probability level.
To evaluate CES with a wide range of applications, we consider inference based on a regression framework.
However, as pointed out by \cite{gneiting2011making}, the stand-alone regression for CES is infeasible since it cannot be represented as the minimizer of an expected loss. To overcome this problem, we adopt the idea in \citet{fissler2016higher} and employ a joint regression framework that simultaneously models the conditional quantile and CES.

For ease of presentation, let $\mathbf{X}$ denote the design vectors for both quantile and ES regression models, but we should bear in mind that one can consider different design vectors for these two models.
For a fixed probability level $\tau \in (0, 1)$, we jointly model the conditional quantile and CES of $Y$ given $\mathbf{X}$ as
\begin{align*} 
Q_\tau(Y | \mathbf{X}) = \mathbf{X}^\prime \boldsymbol{\theta}^q_0 \quad
\text{and} \quad ES_\tau(Y | \mathbf{X}) = \mathbf{X}^\prime \boldsymbol{\theta}^e_0,
\end{align*}
where the parameter vector $\boldsymbol{\theta}_0 = (\boldsymbol{\theta}^{q \prime}_0, \boldsymbol{\theta}^{e \prime}_0)^\prime$ is $\tau$-specific. Denote $u^q(\tau) = Y - Q_\tau(Y | \mathbf{X})$ and $u^e(\tau) = Y - ES_\tau(Y | \mathbf{X})$, we assume $Q_\tau(u^q | \mathbf{X}) = ES_\tau(u^e | \mathbf{X}) = 0$ for identifiability purpose.

To obtain the estimated regression coefficients, we utilize the class of strictly consistent joint loss functions for the pair of quantile and ES \citep{fissler2016higher},
\begin{align}\label{equ:joint_loss}
\rho_\tau(Y, \mathbf{X}, \boldsymbol{\theta}) &= \big\{I(Y \le \mathbf{X}^\prime \boldsymbol{\theta}^q) - \tau\big\} \cdot G_1(\mathbf{X}^\prime \boldsymbol{\theta}^q) - I(Y \le \mathbf{X}^\prime \boldsymbol{\theta}^q) \cdot G_1(Y) \nonumber \\
&+ G_2(\mathbf{X}^\prime \boldsymbol{\theta}^e) \cdot \left\{\mathbf{X}^\prime \boldsymbol{\theta}^e - \mathbf{X}^\prime \boldsymbol{\theta}^q + \frac{(\mathbf{X}^\prime \boldsymbol{\theta}^q - Y) I(Y \le \mathbf{X}^\prime \boldsymbol{\theta}^q)}{\tau}\right\} \\
&- \mathcal{G}_2(\mathbf{X}^\prime \boldsymbol{\theta}^e) + a(Y), \nonumber
\end{align}
where $G_1$ is an increasing and twice continuously differentiable function, $\mathcal{G}_2$ is a three-times continuously differentiable function, $\mathcal{G}_2^{(1)} = G_2$, $G_2$ and $G_2^{(1)}$ are strictly positive, and $G_1$ and $a$ are integrable functions.
\cite{fissler2016higher} also showed that, under some regularity conditions, there exist no strictly consistent loss functions outside the class of functions given above, which implies that \eqref{equ:joint_loss} is the most general class of objective functions that can be applied for the joint regression model.
Given data $(Y_i, \mathbf{X}_i)$, the corresponding joint estimators $\tilde{\boldsymbol{\theta}} = (\tilde{\boldsymbol{\theta}}^{q \prime}, \tilde{\boldsymbol{\theta}}^{e \prime})^\prime$ can be obtained by
\begin{equation}\label{equ:joint_estimation}
\tilde{\boldsymbol{\theta}} = \underset{\boldsymbol{\theta}}{\arg \min} \sum_{i=1}^{n}\rho_\tau(Y_i, \mathbf{X}_i, \boldsymbol{\theta}).
\end{equation}
\cite{dimitriadis2019joint} utilized the joint loss functions in a regression setup and proposed to estimate the quantile and ES parameters jointly by \eqref{equ:joint_estimation}.
However, As a one-step procedure, the estimation is computationally challenging especially when the dimension of parameters is large since the loss functions are neither differentiable nor convex. To reduce the computational effort, we propose a two-step estimation procedure by estimating the quantile parameters first using standard quantile regression.

\begin{remark}
     Following \cite{gneiting2011making} and \cite{fissler2016higher}, we introduce the concept of strictly consistent loss functions. A statistical functional, such as the mean or the $\tau$th quantile, is called elicitable if there is a loss function such that the functional is the unique minimizer of the expected loss. Such a loss function is said to be strictly consistent for the functional. The strictly consistency of $\rho_\tau(Y, \mathbf{X}, \boldsymbol{\theta})$ in \eqref{equ:joint_loss} implies that the parameter vector $\boldsymbol{\theta}_0 = (\boldsymbol{\theta}^{q \prime}_0, \boldsymbol{\theta}^{e \prime}_0)^\prime$ is the unique minimizer of $\rho_\tau(Y, \mathbf{X}, \boldsymbol{\theta})$, and thus $\rho_\tau(Y, \mathbf{X}, \boldsymbol{\theta})$ can be used as the objective function to estimate the regression coefficients.
\end{remark}

\subsection{Two-step estimation}\label{subsec:two_step_estimation_ES}

The first part of the joint loss functions \eqref{equ:joint_loss} corresponds to quantile and depends only on quantile. If the quantile parameters are known by ``oracle", we can plug the quantile parameters into the joint loss functions and use only the second part for the estimation of ES parameters, thus effectively reducing the computational complexity.
Let $\hat{\boldsymbol{\theta}}^q$ be a consistent estimator of $\boldsymbol{\theta}^q_0$, then the first part in \eqref{equ:joint_loss} is fixed given the quantile estimate. In addition, the function $a$ depends only on $Y$ and does not affect the estimation procedure. Therefore, in practice, we can consider the following much simpler plug-in objective functions to obtain the estimated ES regression parameters, 
\begin{align} \label{equ:sep_loss}
\rho_\tau(Y, \mathbf{X}, \hat{\boldsymbol{\theta}}^q, \boldsymbol{\theta}^e) &= G_2(\mathbf{X}^\prime \boldsymbol{\theta}^e) \cdot \left\{\mathbf{X}^\prime \boldsymbol{\theta}^e - \mathbf{X}^\prime \hat{\boldsymbol{\theta}}^q + \frac{(\mathbf{X}^\prime \hat{\boldsymbol{\theta}}^q - Y) I(Y \le \mathbf{X}^\prime \hat{\boldsymbol{\theta}}^q)}{\tau}\right\} - \mathcal{G}_2(\mathbf{X}^\prime \boldsymbol{\theta}^e),
\end{align}
and the associated ES coefficient estimator $\hat{\boldsymbol{\theta}}^e$ can be represented as
\begin{align}\label{equ:two-step_estimator}
\hat{\boldsymbol{\theta}}^e =  \underset{\boldsymbol{\theta}^e}{\arg \min} \sum_{i=1}^n \rho_\tau(Y_i, \mathbf{X}_i, \hat{\boldsymbol{\theta}}^q, \boldsymbol{\theta}^e).
\end{align}

In practice, we estimate $\hat{\boldsymbol{\theta}}^q$ with the \textit{quantreg} package in R. Numerical estimation methods for linear quantile regression has be well developed and the regression coefficients can be obtained efficiently based on linear programming; see \cite{koenkerquantile} for details on linear quantile regression specification. 
Compared with the joint estimator $\tilde{\boldsymbol{\theta}}^e$, the two-step estimator $\hat{\boldsymbol{\theta}}^e$ is computationally more efficient. Furthermore, under the following regularity assumptions, we can show that the two estimators are asymptotically equivalent.

\begin{enumerate}
	\item [A1] The matrix $E\big(\mathbf{X} \mathbf{X}^\prime\big)$ is positive definite.
	\item [A2] The data $(Y_i, \mathbf{X}_i)$ is an independent and identically distributed (i.i.d.) sample of size $n$. Furthermore, the conditional distribution of $Y$ given $\mathbf{X}$, $F_{Y|\mathbf{X}} (\cdot)$ has finite second moment, and is absolutely continuous with a continuous density $f_{Y|\mathbf{X}}$,
	which is strictly positive, continuous and bounded in a neighborhood of the $\tau$th conditional quantile of $Y$.
	\item [A3] The class of strictly consistent joint loss functions is given by \eqref{equ:joint_loss}, where $G_1$ is an increasing and twice continuously differentiable function, $\mathcal{G}_2$ is a three-times continuously differentiable function, $\mathcal{G}_2^{(1)} = G_2$, $G_2$ and $G_2^{(1)}$ are strictly positive, and $G_1$ and $a$ are integrable functions.
	\item [A4] $\hat{\boldsymbol{\theta}}^q$ is a $\sqrt{n}$-consistent estimator of $\boldsymbol{\theta}^q_0$. 
\end{enumerate}

\begin{theorem} \label{thm:sep_asy_normality}
	 Under Assumptions A1, A2, A3, A4 and the Moment Conditions ($\mathcal{M}$-1) in Section \ref{subsec:moment_cond}, we have
	\begin{equation}\label{equ:sep_asy_normality}
	\sqrt{n} \big(\hat{\boldsymbol{\theta}}^e - \boldsymbol{\theta}^e_0\big) \overset{d}{\to} N(\mathbf{0}, \Lambda^{-1} \Omega \Lambda^{-1}),
	\end{equation}
	where $\boldsymbol{\theta}_0 = (\boldsymbol{\theta}^{q \prime}_0, \boldsymbol{\theta}^{e \prime}_0)^\prime$ is the true parameter vector and 
	\begin{align}\label{equ:joint_asy_cov1}
	&\Lambda = E\big\{(\mathbf{X}\mathbf{X}^\prime) G_2^{(1)}(\mathbf{X}^\prime \boldsymbol{\theta}^e_0)\big\}, \\
	\label{equ:joint_asy_cov2}
	&\Omega = E\left[\big(\mathbf{X}\mathbf{X}^\prime\big) \big\{G_2^{(1)}(\mathbf{X}^\prime \boldsymbol{\theta}^e_0)\big\}^2 \times \big\{\frac{1}{\tau} \psi\big(u^q\big) + \frac{1-\tau}{\tau} \phi^2(\mathbf{X}, \boldsymbol{\theta}^q_0, \boldsymbol{\theta}^e_0)\big\}\right], \\
	\label{equ:joint_asy_cov3}
	&\psi\big(u^q\big) = \text{Var}\big(u^q | u^q \le 0, \mathbf{X}\big) = \text{Var}\big(Y - \mathbf{X}^\prime \boldsymbol{\theta}^q_0 | Y \le \mathbf{X}^\prime \boldsymbol{\theta}^q_0, \mathbf{X}\big), \\
	\label{equ:joint_asy_cov4}
	&\phi(\mathbf{X}, \boldsymbol{\theta}^q_0, \boldsymbol{\theta}^e_0) = \mathbf{X}^\prime \boldsymbol{\theta}^q_0 - \mathbf{X}^\prime \boldsymbol{\theta}^e_0.
	\end{align}
\end{theorem}

Assumption A1 is required to exclude the multicollinearity of the stochastic explanatory variables. Assumption A2 includes standard assumptions in mean and quantile regression, and the finite conditional moment of $Y$ given $\mathbf{X}$ is assumed since ES is a truncated mean of quantiles. The conditions on functions $G_1$ and $\mathcal{G}_2$ are required for the strictly consistent joint loss functions.
Assumption A4 is made for convenience, and it can be relaxed to $\sqrt{n}\|\hat{\boldsymbol{\theta}}^q - \boldsymbol{\theta}^q_0\|^2 = o_p(1)$; see discussions in the proof of Lemma \ref{lemma:sep_asy_normality_lemma3} in Section \ref{subsec:proof_of_thm}.

\begin{remark}
   \cite{dimitriadis2019joint} established the asymptotic normality of the joint estimators $(\tilde{\boldsymbol{\theta}}^{q}, \tilde{\boldsymbol{\theta}}^{e})$, and show that the two estimators are asymptotically independent. Theorem \ref{thm:sep_asy_normality} implies that the proposed two-step estimator $\hat{\boldsymbol{\theta}}^e$ is asymptotically equivalent to the joint estimator $\tilde{\boldsymbol{\theta}}^{e}$, but the former is obtained with two steps and is numerically more efficient. In addition, the error involved in the quantile estimation $\hat{\boldsymbol{\theta}}^q$ in the first step does not affect the asymptotic distribution of $\hat{\boldsymbol{\theta}}^e$, which agrees with the results of \cite{dimitriadis2019joint}. This asymptotic independence result follows because 
\begin{align*}
    \frac{\partial^2 E\{\rho(Y, \mathbf{X}, \boldsymbol{\theta}^q, \boldsymbol{\theta}^e) | \mathbf{X}\}}{(\partial \boldsymbol{\theta}^q \partial \boldsymbol{\theta}^{e \prime})} \big|_{\boldsymbol{\theta}^q = \boldsymbol{\theta}^q_0} & = (\mathbf{X} \mathbf{X}^\prime) G_2^{(1)}(\mathbf{X}^\prime \boldsymbol{\theta}^e) \frac{F_{Y|\mathbf{X}}(\mathbf{X}^\prime \boldsymbol{\theta}^q) - \tau}{\tau} \big|_{\boldsymbol{\theta}^q = \boldsymbol{\theta}^q_0} = \boldsymbol{0}.
\end{align*}
That is, the partial derivative of the joint loss function \eqref{equ:joint_loss} evaluated at the true quantile coefficient is zero.
Therefore, when $\sqrt{n}\|\hat{\boldsymbol{\theta}}^q - \boldsymbol{\theta}^q_0\|^2 = o_p(1)$, $\hat{\boldsymbol{\theta}}^e$ is locally robust to the prior quantile estimation or its local misspecification; see \cite{chernozhukov2016locally} for an elaboration on local robustness. Furthermore, even though we assume both linear models for quantile and ES regression, the local robustness property enables us to consider more general models for the quantile estimation in the first step. For instance, the conditional quantile in the first step can be obtained by nonparametric regression, and this will not affect the asymptotic property of the two-step estimator $\hat{\boldsymbol{\theta}}^e$ as long as the ES regression model is correctly specified and the conditional quantile estimation is consistent with a certain rate.
\end{remark}

Based on the asymptotic normality of the two-step estimator $\hat{\boldsymbol{\theta}}^e$, 
a Wald-type test can be constructed for inference on $\boldsymbol{\theta}^e_0$ through direct estimation of the covariance matrix, which involves the conditional variance of the quantile residuals $\psi(u^q)$ given in \eqref{equ:joint_asy_cov3}.
However, accurate estimation of this nuisance quantity is challenging. First of all, for tail quantile levels, e.g., $\tau$ close to 0, corresponding to the left tail,
there exists very few (about $n \cdot \tau$) observations after conditional on $u^q \le 0$. Moreover, taking the dependence of the covariates $\mathbf{X}$ into consideration further complicates the estimation of the conditional variance, especially when the sample size is small.

\cite{dimitriadis2019joint} provided four different ways to directly estimate the asymptotic variance of the joint estimator $\tilde{\boldsymbol{\theta}}^{e}$. Since the joint estimator $\tilde{\boldsymbol{\theta}}^{e}$ and our proposed two-step estimator $\hat{\boldsymbol{\theta}}^{e}$ are asymptotically equivalent, we can adopt the same variance estimation procedures in \cite{dimitriadis2019joint}. We summarize different methods below.
The first three approaches involve the estimation of the nuisance quantity $\psi(u^q) = \text{Var}(u^q | u^q \le 0, \mathbf{X})$.
When errors are homogeneous such that the distribution of $u^q$ is independent of the regression covariates $\mathbf{X}$, we can estimate $\psi(u^q)$ simply by the sample variance of the negative residuals, that is,
\begin{align*}
    \hat{\psi}_I = \text{Var}(\hat{u}^q | \hat{u}^q \le 0),
\end{align*}
where $\hat{u}^q = Y - \mathbf{X}^\prime \hat{\boldsymbol{\theta}}^q$ is the estimated quantile residual, and we refer to $\hat{\psi}_I$ as the \textit{iid} estimator. The second estimator allows for a location-scale dependence structure of the quantile residuals on $\mathbf{X}$,
\begin{equation}\label{equ:ls_quant_res}
u^q = \mathbf{X}^\prime \boldsymbol{\alpha} + \big(\mathbf{X}^\prime \boldsymbol{\nu}\big) \cdot \epsilon,
\end{equation}
where $\boldsymbol{\alpha}$ and $\boldsymbol{\nu}$ are $p$-dimensional parameter vectors, and $\epsilon \sim F_\epsilon(0, 1)$ follows some distribution $F_\epsilon$ with zero-mean and unit variance. The conditional distribution of $u^q$ given $\mathbf{X}$ is $F_\epsilon\big(\mathbf{X}^\prime \boldsymbol{\alpha}, (\mathbf{X}^\prime \boldsymbol{\nu})^2 \big)$, 
and the truncated conditional density of $u^q$ given $u^q \le 0$ and $\mathbf{X}$ is 
\begin{align} \label{equ:h_fun}
    h(z | \mathbf{X}) = (\mathbf{X}^\prime \nu)^{-1} f_\epsilon \big(\frac{z - \mathbf{X}^\prime \alpha}{\mathbf{X}^\prime \nu}\big) \big/ F_\epsilon\big(-\frac{\mathbf{X}^\prime \alpha}{\mathbf{X}^\prime \nu} \big).
\end{align}
The conditional variance $\psi(u^q)$ can be estimated by quasi generalized pseudo maximum likelihood (\citeauthor{gourieroux1984pseudo}, \citeyear{gourieroux1984pseudo}) based on the scaling formula $$
\text{Var}\big(u^q | u^q \le 0, \mathbf{X}\big) = \int_{-\infty}^{0} z^2 h(z | \mathbf{X}) dz - \big(\int_{-\infty}^{0} z h(z | \mathbf{X}) dz\big)^2.
$$ 
In practice, we first estimate the conditional mean and variance of $u^q$ given $\mathbf{X}$ by MLE (maximum likelihood estimator) and then employ kernel density estimation to estimate the unknown distribution $F_\epsilon$ nonparametrically. Then truncated density $h(\cdot | \mathbf{X})$ is calculated by \eqref{equ:h_fun} accordingly. The resulting estimator of $\psi$ is referred to as the \textit{nid} estimator (denoted by $\hat{\psi}_N$). 

The third option discussed in \cite{dimitriadis2019joint} is by assuming that $\epsilon \sim N(0, 1)$ in \eqref{equ:ls_quant_res}. However, our empirical investigation suggests that this approach does not perform well in some situations. Therefore, throughout the numerical studies, we focus on the $\textit{iid}$ and $\textit{nid}$ approaches for estimating the conditional variance $\psi(u^q)$.

Another feasible alternative of covariance estimation is by adopting the \textit{bootstrap} method \citep{efron1992bootstrap}. We generate $B$ bootstrap samples by randomly selecting the $n$ pairs of $(Y_i, \mathbf{X}_i)$ with replacement. We can then obtain $B$ bootstrap ES coefficient estimators to each of the bootstrap samples by applying either the one-step or the proposed two-step estimation approach.
The bootstrap covariance is then approximated by the sample covariance of the $B$ bootstrap parameter estimates.

\begin{theorem} \label{thm:joint_var_consistency}
Let $\hat{\boldsymbol{\theta}}^e$ be the two-step estimator obtained by \eqref{equ:two-step_estimator} and denote $\hat{u}^q_i = Y_i - \mathbf{X}_i^\prime \hat{\boldsymbol{\theta}}^q$,
	\begin{align*}
	&\hat{\Lambda} = n^{-1} \sum_{i=1}^n (\mathbf{X}_i \mathbf{X}_i^\prime) \cdot G_2^{(1)} (\mathbf{X}_i^\prime \hat{\boldsymbol{\theta}}^e), \\
	&\hat{\Omega}(\hat{\psi}) = n^{-1} \sum_{i=1}^n \big(\mathbf{X}_i \mathbf{X}_i^\prime\big) \cdot \big\{G_2^{(1)} (\mathbf{X}_i^\prime \hat{\boldsymbol{\theta}}^e)\big\}^2 \cdot \left\{\frac{1}{\tau} \hat{\psi}(\hat{u}^q_i) + \frac{1-\tau}{\tau}\phi^2(\mathbf{X}_i, \hat{\theta}^q, \hat{\theta}^e)\right\}.
	\end{align*}
	In addition, denote $\hat{\Omega}(\hat{\psi}_I)$ and $\hat{\Omega}(\hat{\psi}_N)$ as the estimated covariance matrices with the conditional variance $\psi(u^q) = \text{Var}(u^q | u^q \le 0, \mathbf{X})$ estimated by $\hat{\psi}_I$ and $\hat{\psi}_N$, respectively. If the assumptions of Theorem  \ref{thm:sep_asy_normality} hold, then $\hat{\Lambda} - \Lambda = o_p(1)$ and $\hat{\Omega}(\hat{\psi}_N) - \Omega = o_p(1)$. Furthermore, if $u^q$ is independent with $\mathbf{X}$, then $\hat{\Omega}(\hat{\psi}_I) - \Omega = o_p(1)$.
\end{theorem}

\subsection{Proposed score test}\label{subsec:proposed_score_test_ES}
Theorem \ref{thm:joint_var_consistency} shows that the covariance matrix of $\hat{\boldsymbol{\theta}}^e$ can be estimated consistently.
However, in the quantile regression literature, it has been shown that Wald-type test based on direct estimation of the asymptotic covariance matrix is often unstable for small sample sizes.
One reason is due to the sensitivity of Wald-type test to the smoothing parameter involved in estimating the unknown conditional density function. The score test has been shown to have more stable performance than Wald test for quantile regression in finite samples \citep{chen2005computational, kocherginsky2005practical}.
Due to the connection between quantile and ES, we propose an alternative score-type test for the inference on ES regression parameters.

We partition $\boldsymbol{\theta}^e$ into two parts $\boldsymbol{\theta}_1^e \in \mathbb{R}^{p_1}$ and $\boldsymbol{\theta}_2^e \in \mathbb{R}^{p_2}$ with $p_1 + p_2 = p$, and let $\mathbf{W}$ and $\mathbf{Z}$ be the design vectors corresponding to $\boldsymbol{\theta}_1^e$ and $\boldsymbol{\theta}_2^e$, respectively. Suppose we want to test the hypotheses $H_0: \boldsymbol{\theta}_2^e = \mathbf{0}$ against $H_a: \boldsymbol{\theta}_2^e \ne \mathbf{0}$ in the joint regression model
\begin{equation*}
Q_\tau(Y | \mathbf{X}) = \mathbf{X}^\prime \boldsymbol{\theta}^q \quad \text{and} \quad ES_\tau(Y | \mathbf{X}) = \mathbf{W}^\prime \boldsymbol{\theta}^e_1 + \mathbf{Z}^\prime \boldsymbol{\theta}^e_2.
\end{equation*}

Denote
\begin{align} \label{equ:ortho_score}
&\Pi_W = (\mathbf{W}_1, \dots, \mathbf{W}_n)^\prime, \quad \Pi_Z = (\mathbf{Z}_1, \dots, \mathbf{Z}_n)^\prime, \nonumber \\
& G = \text{diag}\{G_2^{(1)}(\mathbf{X}_1^\prime \boldsymbol{\theta}^e), \dots, G_2^{(1)}(\mathbf{X}_n^\prime \boldsymbol{\theta}^e)\},  \\
&P = \Pi_W(\Pi_W^\prime G \Pi_W)^{-1} \Pi_W^\prime G, \qquad
\Pi_Z^* = (I - P)\Pi_Z, \nonumber
\end{align} 
and let $\mathbf{Z}_i^{*\prime}$ be the rows of $\Pi_Z^*$ corresponding to the $i$th subject. We consider the orthogonal transformation on $\mathbf{Z}$ to adjust for the dependence of $\mathbf{Z}$ and $\mathbf{W}$. This transformation is needed to cancel out the first-order bias involved in $\hat{\boldsymbol{\theta}}^e_1$ to prove Lemma \ref{lemma:score_stat_null_lemma2} in Section \ref{subsec:proof_of_thm}. The weighted projection in the orthogonal transformation through $G$ is needed to
to account for the heteroscedasticity; see some related discussion under quantile regression in \cite{koenker1999goodness}.
Our proposed score test statistic is defined as
\begin{equation}\label{equ:score_stat}
T_n(\hat{\psi}) = S_n^\prime \big\{\hat{\Sigma}_n(\hat{\psi})\big\}^{-1} S_n,
\end{equation}
where
\begin{align*}
&S_n = n^{-1/2}\sum_{i=1}^n \hat{\mathbf{Z}}_i^* G_2^{(1)}(\mathbf{X}_i^\prime \hat{\boldsymbol{\theta}}^e) \big\{\mathbf{W}_i^\prime \hat{\boldsymbol{\theta}}_1^e - \mathbf{X}_i^\prime \hat{\boldsymbol{\theta}}^q - \tau^{-1} \hat{u}^q_i I(\hat{u}^q_i \le 0)\big\}, \quad \hat{u}^q_i = Y_i - \mathbf{X}_i^\prime \hat{\boldsymbol{\theta}}^q,  \\
&\hat{\Sigma}_n(\hat{\psi}) = n^{-1} \sum_{i=1}^n \big(\hat{\mathbf{Z}}_i^* \hat{\mathbf{Z}}_i^{*\prime}\big) \big\{G_2^{(1)}(\mathbf{X}_i^\prime \hat{\boldsymbol{\theta}}^e)\big\}^2 \cdot \left\{\frac{1}{\tau} \hat{\psi}(\hat{u}^q_i) + \frac{1-\tau}{\tau} \phi^2(\mathbf{X}_i, \hat{\boldsymbol{\theta}}^q, \hat{\boldsymbol{\theta}}^e)\right\}.
\end{align*}
Here $\hat{\boldsymbol{\theta}}_1^e$ is the two-step estimator of the ES regression parameter $\boldsymbol{\theta}^{e}_1$ under $H_0$, and $\hat{\boldsymbol{\theta}}^e$ is the two-step estimator of $\boldsymbol{\theta}^e = (\boldsymbol{\theta}^{e \prime}_1, \boldsymbol{\theta}^{e \prime}_2)^\prime$ under the unrestricted model. To obtain $\hat{\mathbf{Z}}_i^*$, the weight matrix $G$ given in \eqref{equ:ortho_score} can be estimated with $\hat{\boldsymbol{\theta}}^e$.

Before presenting the asymptotic distribution of $T_n$, we
define $$
	\Sigma = E\left[\big(\mathbf{Z}^* \mathbf{Z}^{*\prime}\big) \big\{G_2^{(1)}(\mathbf{X}^\prime \boldsymbol{\theta}_{0}^e)\big\}^2\left\{\frac{1}{\tau} \psi(u^q) + \frac{1-\tau}{\tau} \phi^2(\mathbf{X}, \boldsymbol{\theta}_0^q, \boldsymbol{\theta}_0^e)\right\}\right].
	$$
We also introduce an additional assumption A5.
\begin{itemize}
	\item [A5] The minimum eigenvalue of $\Sigma$ is bounded away from zero.
\end{itemize}

\begin{theorem}\label{thm:score_stat_null}
	Suppose that assumptions in Theorem \ref{thm:sep_asy_normality} and A5 hold, we have:
	\begin{itemize}
	    \item [(i)] under $H_0$, $T_n(\hat{\psi}_N) \overset{d}{\to} \chi^2_{p_2}$ as $n \to \infty$. Furthermore, if $u^q$ is independent with $\mathbf{X}$, $T_n(\hat{\psi}_I) \overset{d}{\to} \chi^2_{p_2}$ as $n \to \infty$;
	    \item [(ii)] under the local alternative hypothesis $H_n: \boldsymbol{\theta}_2^e = \boldsymbol{\theta}_{20}^e / \sqrt{n}$ with $\boldsymbol{\theta}_{20}^e$ be some non-zero parameter vector corresponding to $\mathbf{Z}$, $T_n(\hat{\psi}_N)$ asymptotically follows a non-central $\chi^2_{p_2}$ distribution with the noncentrality parameter $$
	    \zeta = E\big\{\mathbf{Z}^* G_2^{(1)}(\mathbf{X}^\prime \boldsymbol{\theta}_{0}^e) (\mathbf{Z}^\prime \boldsymbol{\theta}_{20}^e)\big\}^\prime \Sigma^{-1} E\big\{\mathbf{Z}^* G_2^{(1)}(\mathbf{X}^\prime \boldsymbol{\theta}_{0}^e) (\mathbf{Z}^\prime \boldsymbol{\theta}_{20}^e)\big\}.
	    $$
	    Furthermore, if $u^q$ is independent with $\mathbf{X}$, then $T_n(\hat{\psi}_I) \overset{d}{\to} \chi^2_{p_2}$ with the noncentrality parameter $\zeta$.
	\end{itemize}
\end{theorem}

\begin{remark}
	 The matrix $\hat{\Sigma}_n(\hat{\psi})$ in the score test statistic involves the estimation of the truncated conditional variance of $u^q$ given $u^q \le 0$ and $\mathbf{X}$. Similar to the Wald-type test, we consider both $\textit{iid}$ and $\textit{nid}$ estimators for the nuisance parameter $\psi$ to accommodate different scenarios.
\end{remark}

\begin{remark}
     The proposed estimation and inference methods can be directly applied for analyzing the upper tail CES at the probability level $\tau$, which is defined as
     \begin{equation} \label{equ:upper_CES}
         ES_\tau(Y | \mathbf{X}) = (1-\tau)^{-1}\int_\tau^1 F^{-1}_{Y|\mathbf{X}}(u)du = (1-\tau)^{-1}\int_\tau^1 Q_u(Y | \mathbf{X}) du.
     \end{equation}
     Assuming the upper tail ES regression model:
     $$
     ES_\tau(Y | \mathbf{X}) = \mathbf{X}^\prime \boldsymbol{\theta}^{e \star}.
     $$
     That is, $\boldsymbol{\theta}^{e \star}$ and $\boldsymbol{\theta}^{e}$ are the upper and lower tail ES regression parameters, respectively.
     Note that $F^{-1}_{Y|\mathbf{X}}(u) = -F^{-1}_{-Y|\mathbf{X}}(1-u)$, which implies the $u$th conditional quantile of $Y$ given $\mathbf{X}$ is the negative of the $(1-u)$th conditional quantile of $-Y$ given $\mathbf{X}$.
     In practice, if we are interested in the inference on $\boldsymbol{\theta}^{e \star}$, we can (1) change $Y$ to $Y^\star = -Y$ and let $\tau^\star = 1 - \tau$; (2) apply the proposed two-step estimation and inference methods on the lower tail ES of $Y^\star$ conditional on $\mathbf{X}$ with probability level $\tau^\star$. Then we have $\hat{\boldsymbol{\theta}}^{e \star} = -\hat{\boldsymbol{\theta}}^{e}$ and $\text{Var}(\hat{\boldsymbol{\theta}}^{e \star}) = \text{Var}(\hat{\boldsymbol{\theta}}^{e})$.
\end{remark}

\section{Simulation study}\label{sec:simulation_ES}
In this section, we investigate the finite sample performance of the proposed inference method through Monte Carlo simulation studies. For comparison purpose, we include the results given by the Wald-type and \textit{bootstrap} methods, which are based on the joint estimation and are implemented in the R package \textit{esreg} \citep{dimitriadis2019joint}. 
For both Wald and score methods, we consider \textit{W-IID} and \textit{S-IID} approaches, where the asymptotic variance is estimated under the homogeneous error assumption; and \textit{W-NID} and \textit{S-NID} method, where the conditional variance $\psi(u^q)$ is estimated by $\hat{\psi}_N$.
In addition, we also report the results of the \textit{COVES} test introduced by \cite{he2010detection}. Since the \textit{COVES} method focuses on the treatment difference at the right tail of the response distribution, for the simulation study in Section \ref{sec:simulation_ES} and the real data analysis in Section \ref{sec:real_data_ES}, we will focus on the inference for the upper tail CES at the probability level $\tau$, as defined in \eqref{equ:upper_CES}.

\subsection{Simulation Design}\label{subsec:simulation_design_ES}
The first two models we consider have simple setups:
\begin{itemize}
	\item \textit{Scenario 1}: $Y = 5 + \eta D + x_1 + \epsilon$; 
	\item \textit{Scenario 2}: $Y = 5 - \eta D + C + (1 + 0.25D + 2C)\epsilon$;
\end{itemize}
where $D$ is the binary treatment indicator, $x_1 \sim N(2.5, 0.5^2)$ and $\epsilon \sim N(0, 1)$, $C$ is an unbalanced covariate with $C \sim TN(\min = -0.5, \mu = 0.5, \sigma^2 = 0.5^2)$ (truncated normal distribution) in the treatment group and $C \sim TN(\min = -0.5, \mu = 0, \sigma^2 = 0.5^2)$ in control group. 

The first scenario is a homogeneous model where the regression error does not interact with any covariates. For the second model, the error depends on both the treatment variable and the unbalanced covariate $C$. At the probability level $\tau$, the marginal impact on the CES due to treatment effect is $$
\eta(\tau) = -\eta + 0.25 \cdot ES_\tau(\epsilon),$$
which is the ES coefficient associated with the treatment variable $D$.
In contrast, the expectation of the \textit{COVES} test statistic is given by \begin{equation}\label{equ:coves_stat}
    \mathcal{T}_\tau = \eta(\tau) + 2\big\{E(C|D=1) - E[(C|D=0)\big\} \cdot \big\{ES_\tau(\epsilon) - Q_\tau(\epsilon)\big\},
\end{equation}
where $Q_\tau(\epsilon)$ and $ES_\tau(\epsilon)$ are the marginal $\tau$th quantile and upper tail ES of $\epsilon$.
Due to the imbalance of covariate $C$, \textit{COVES} may inflate or deflate the treatment difference $\eta(\tau)$, which consequently makes the test either too liberal or too conservative.

\begin{remark}
     The difference between $\mathcal{T_\tau}$ and the treatment difference $\eta(\tau)$ is $$
     \mathcal{T_\tau} - \eta(\tau) = 2 \big\{E(C|D=1) - E[(C|D=0)\big\} \cdot \big\{ES_\tau(\epsilon) - Q_\tau(\epsilon)\big\}.
     $$
     Suppose that $C$ is an unbalanced covariate such that the mean of $C$ differs for the two treatment groups and the error $\epsilon$ depends on $C$. Then it is possible to have $\mathcal{T_\tau} - \eta(\tau) \ne 0$, and this may affect the power of the \textit{COVES} test.
     Specifically, if the second term on the right-hand-side (RHS) of \eqref{equ:coves_stat} cancels out with $\eta(\tau)$, \textit{COVES} may fail to detect the treatment difference.
     On the other hand, if the treatment has no impact on CES such that $\eta(\tau) = 0$, but the second term on the RHS of \eqref{equ:coves_stat} is non-zero, then \textit{COVES} may over-reject the null hypothesis and thus lead to higher false positive rate.
\end{remark}

There's only one confounding variable in the first two models, and both error terms follow normal distributions. To examine the robustness of the proposed method, we further consider another two scenarios where more regression covariates are included (Scenario 3) and the error has a heavy-tailed distribution (Scenario 4). The data are generated from
\begin{equation*}\label{equ:simeq3}
Y = 5 + \eta D + \sum_{i=2}^n x_i + (1 + \gamma D)\epsilon,
\end{equation*}
where $x_2$ is $Ber(0.4)$, $x_3$ and $x_4$ have standard log-normal distribution, $(x_5, x_6)$ is bivariate normal with mean (2,2), variance (1,1), and correlation 0.8, $x_7$ is chi-square distributed with one degree of freedom. Except for the correlation between $x_5$ and $x_6$, all other variables are independently generated.

\begin{itemize}
	\item \textit{Scenario 3}: we take $\gamma = 0$, and the error term $\epsilon \sim N(0, 1)$.
	\item \textit{Scenario 4}: we take $\gamma = 0.2$, and the error term $\epsilon \sim t_3/2$.
\end{itemize}

We consider two sample sizes $n = 50$ and $n = 100$ for each treatment group, and we focus on $\tau = 0.8$ and $\tau = 0.9$ in this study. The simulation is repeated 600 times for each scenario with a given value of $\eta$.

\subsection{Statistical Power for Testing the Treatment Effect}\label{subsec:power_ES}
For both Wald-type and score-type approaches, the estimation efficiency depends on the form of specification functions $G_1$ and $\mathcal{G}_2$.
\cite{dimitriadis2019joint} discusses several feasible choices, and their simulation analysis suggested that $G_1(z) = z$ and $\mathcal{G}_2(z) = -\log(-z)$ provide the most consistent estimation results under all scenarios considered. Throughout, we'll employ these two functions for the regression procedure.

Table \ref{tab:T1E} shows that Wald-type and score-type approaches both maintain the significance level reasonably well and the corresponding type I errors stay close to the nominal level of 0.05. However, the bootstrap method and the \textit{COVES} test yield inflated false positive rates for most cases, especially when sample size $n = 50$ and $\tau = 0.9$.
Under scenario 1, with only one covariate besides the treatment indicator, all methods perform quite similarly. However, as we add more confounding factors, the score-type testing methods show higher statistical power than the Wald-type approaches; see Figure \ref{fig:PA_plots} for some typical examples in Scenarios 3 and 4 at $\tau = 0.8$ and $n = 100$. And the power curves in Scenario 2 confirm that \textit{COVES} test gives biased estimation of the treatment difference due to the unbalanced covariate.

\begin{table}[H]
	\centering
	\caption{Type I error (percentage) for testing $H_0: \eta(\tau) = 0$ with nominal level of 5\%.}
	\label{tab:T1E}
	\begin{tabular}{ccccccccc}
		\hline
		Scenario & $n$ & $\tau$ & \textit{W-IID} &  \textit{W-NID} & \textit{S-IID} & \textit{S-NID} &  \textit{BOOT} & \textit{COVES} \\
		\hline
		\multirow{4}{*}{1}& \multirow{2}{*}{50} & 0.8 & 7.3 & 6.7 & 6.2 & 6.3 & 6.7 & 8.5 \\  
		&& 0.9 & 9.3 & 8.5 & 9.7 & 8.0 & 8.0 & \textbf{13.2} \\ 
		&\multirow{2}{*}{100}& 0.8 & 6.5 & 6.0 & 6.2 & 6.2 & 6.3 & 6.7 \\ 
		&& 0.9 & 6.7 & 5.5 & 6.2 & 5.5 & 6.7 & 7.7 \\ 
		\hline
		\multirow{4}{*}{2}& \multirow{2}{*}{50} & 0.8 & 2.3 & 3.8 & 2.7 & 4.2 & 6.8 & \textbf{13.0} \\ 
		&& 0.9 & 4.7 & 6.2 & 5.8 & 7.0 & 7.0 & \textbf{12.5} \\ 
		& \multirow{2}{*}{100} & 0.8 & 1.2 & 2.5 & 1.7 & 3.0 & 5.5 & \textbf{18.3} \\ 
		&& 0.9 & 2.0 & 3.0 & 2.3 & 3.2 & 5.3 & \textbf{11.8} \\ 
		\hline
		\multirow{4}{*}{3}& \multirow{2}{*}{50} & 0.8 & 5.7 & 6.0 & 4.5 & 4.8 & 8.2 & \textbf{12.5} \\ 
		&& 0.9 & 8.2 & 8.2 & 8.5 & 8.2 & \textbf{10.0} & \textbf{21.8} \\  
		& \multirow{2}{*}{100} & 0.8 & 3.5 & 4.3 & 4.0 & 4.3 & 5.2 & 6.7 \\ 
		&& 0.9 & 4.0 & 4.0 & 5.2 & 4.8 & 6.2 & 9.8 \\ 
		\hline
		\multirow{4}{*}{4}& \multirow{2}{*}{50} & 0.8 & 2.5 & 3.0 & 2.0 & 3.0 & 6.0 & 9.5 \\ 
		&& 0.9 & 3.0 & 3.5 & 4.7 & 5.0 & \textbf{11.0} & \textbf{17.8} \\ 
		& \multirow{2}{*}{100} & 0.8 & 3.2 & 3.7 & 1.8 & 2.7 & 7.7 & 5.5 \\ 
		&& 0.9 & 2.7 & 3.3 & 3.2 & 4.0 & 9.5 & \textbf{10.7} \\ 
		\hline
	\end{tabular}
	\begin{itemize}
	    \item [] \footnotesize \textit{W-IID} (\textit{W-NID}): Wald-type methods with $\psi(u^q)$ estimated by $\hat{\psi}_I$ ($\hat{\psi}_N$); \textit{S-IID} (\textit{S-NID}): score methods with $\psi(u^q)$ estimated by $\hat{\psi}_I$ ($\hat{\psi}_N$); \textit{BOOT}: bootstrap method based on the joint estimation; \textit{COVES}: method in \cite{he2010detection}.
	\end{itemize}
\end{table}

\subsection{Comparing the Confidence Interval of Treatment Coefficient}\label{subsec:CI_ES}
To access the performance of different methods for confidence interval construction, we fix $\eta = 1.35, 2, 2.5$ and 3.5 in Scenarios 1-4 respectively.
Tables \ref{tab:CI_re_50} and \ref{tab:CI_re_100} summarize the coverage percentage and average length of 95\% confidence intervals for the treatment difference $\eta(\tau)$. Under Scenario 1, the Wald-type and score methods show similar accuracy. For Scenarios 3 and 4 with more confounding factors, the score-type methods provide shorter confidence intervals with relatively higher coverage, which agrees with the results of the power analysis in Section \ref{subsec:power_ES}. Furthermore, when errors are i.i.d,
\textit{W-IID} and \textit{W-NID} approaches perform similarly. For Scenarios 2 and 4 when the errors are heterogeneous, \textit{W-NID} method shows better performance in the sense that it provides confidence intervals with coverage closer to the nominal level and shorter length than the \textit{W-IID} method.
On the other hand, the score methods are less sensitive to the violation of homogeneity assumption, as \textit{S-IID} and \textit{S-NID} give similar results across all scenarios considered. Overall, the \textit{COVES} method gives lowest coverage percentage, much below the nominal level 95\%, especially under Scenario 2 where the error term depends on an unbalanced covariate.

\begin{table}[H]
	\centering
	\caption{Coverage probabilities and average lengths (inside the parentheses) of $95\%$ confidence intervals from different methods when sample size $n = 50$. All values are in percentages.}
	\label{tab:CI_re_50}
	\begin{tabular}{cccccccc}
		\hline
		Scenario & $\tau$ & \textit{W-IID} &  \textit{W-NID} & \textit{S-IID} & \textit{S-NID} &  \textit{BOOT} & \textit{COVES} \\
		\hline
		\multirow{4}{*}{1}& \multirow{2}{*}{0.8} & 92.2 & 92.8 & 94.0 & 94.0 & 93.0 & 91.8 \\ 
		&& (120) & (122) & (120) & (120) & (124) & (113) \\ 
		& \multirow{2}{*}{0.9} &90.5 & 91.5 & 90.7 & 90.7 & 91.3 & \textcolor{red}{88.3} \\ 
		&& (146) & (149) & (146) & (146) & (144) & (133) \\ 
		\hline
		\multirow{4}{*}{2}& \multirow{2}{*}{0.8}&98.0 & 96.0 & 97.3 & 97.3& 93.5 & \textcolor{red}{85.8} \\ 
		&& (268) & (240) & (265) & (265) & (203) & (232) \\ 
		& \multirow{2}{*}{0.9}  &95.5 & 94.5 & 94.2 & 94.2& 93.3 & \textcolor{red}{84.7} \\ 
		&& (327) & (297) & (320) & (320) & (253) & (268) \\ 
		\hline
		\multirow{4}{*}{3}& \multirow{2}{*}{0.8} &92.7 & 92.3 & 95.5 & 95.5 & 94.7 & \textcolor{red}{87.7} \\ 
		&& (184) & (182) & (\textcolor{blue}{157}) & (\textcolor{blue}{157}) & (187) & (110) \\ 
		&\multirow{2}{*}{0.9} & 90.8 & 91.5 & 91.5 & 91.5 & 92.5 & \textcolor{red}{77.3} \\ 
		&& (210) & (210) & (\textcolor{blue}{175}) & (\textcolor{blue}{175}) & (192) & (119) \\ 
		\hline
		\multirow{4}{*}{4}& \multirow{2}{*}{0.8}&97.0 & 96.3 & 98.0 & 98.0 & 96.5 & 90.7 \\ 
		&& (234) & (227) & (\textcolor{blue}{200}) & (\textcolor{blue}{200}) & (197) & (143) \\ 
		& \multirow{2}{*}{0.9}  &96.3 & 96.2 & 95.5 & 95.5 & 94.3 & \textcolor{red}{83.3} \\ 
		&& (343) & (333) & (\textcolor{blue}{284}) & (\textcolor{blue}{284}) & (238) & (201) \\ 
		\hline
	\end{tabular}
	\begin{itemize}
	    \item [] \footnotesize \textit{W-IID} (\textit{W-NID}): Wald-type methods with $\psi(u^q)$ estimated by $\hat{\psi}_I$ ($\hat{\psi}_N$); \textit{S-IID} (\textit{S-NID}): score methods with $\psi(u^q)$ estimated by $\hat{\psi}_I$ ($\hat{\psi}_N$); \textit{BOOT}: bootstrap method based on the joint estimation; \textit{COVES}: method in \cite{he2010detection}.
	\end{itemize}
\end{table}

\begin{table}[H]
	\centering
	\caption{Coverage probabilities and average lengths (inside the parentheses) of $95\%$ confidence intervals from different methods when sample size $n = 100$. All values are in percentages.}
	\label{tab:CI_re_100}
	\begin{tabular}{cccccccc}
		\hline
		Scenario & $\tau$ & \textit{W-IID} &  \textit{W-NID} & \textit{S-IID} & \textit{S-NID} &  \textit{BOOT} & \textit{COVES} \\
		\hline
		\multirow{4}{*}{1}& \multirow{2}{*}{0.8}&93.3 & 93.8 & 93.8 & 93.8 & 93.7 & 92.3 \\ 
		&& (87.1) & (88.4) & (87.2) & (87.2) & (87.8) & (83.8) \\ 
		& \multirow{2}{*}{0.9}  &93 & 93.2 & 93.5 & 93.5 & 92.5 & 91.5 \\ 
		&& (109) & (110) & (108) & (108) & (107) & (103) \\ 
		\hline
		\multirow{4}{*}{2}& \multirow{2}{*}{0.8}&98.7 & 97.7 & 98.3 & 98.3 & 94.7 & \textcolor{red}{82.7} \\ 
		&& (194) & (173) & (192) & (192) & (143) & (174) \\ 
		& \multirow{2}{*}{0.9}  &98.2 & 97.5 & 97.7 & 97.7 & 94.7 & \textcolor{red}{86.8} \\ 
		&& (242) & (215) & (238) & (238) & (178) & (213) \\ 
		\hline
		\multirow{4}{*}{3}& \multirow{2}{*}{0.8} &94.7 & 94.5 & 96.2 & 96.2 & 95.7 & 93.7 \\ 
		&& (129) & (129) & (\textcolor{blue}{108}) & (\textcolor{blue}{108}) & (128) & (81.4) \\ 
		& \multirow{2}{*}{0.9}  &94.8 & 94.0 & 95.2 & 95.2 & 95.0 & \textcolor{red}{89.5} \\ 
		&& (150) & (149) & (\textcolor{blue}{130}) & (\textcolor{blue}{130}) & (137) & (96.3) \\ 
		\hline
		\multirow{4}{*}{4}& \multirow{2}{*}{0.8}&96.2 & 95.7 & 98.2 & 98.2 & 96.8 & 94.3 \\ 
		&& (170) & (161) & (\textcolor{blue}{152}) & (\textcolor{blue}{152})& (142) & (113) \\ 
		& \multirow{2}{*}{0.9}  &96.2 & 95.3 & 96.8 & 96.8 & 93.7 & \textcolor{red}{89.3} \\ 
		&& (260) & (247) & (\textcolor{blue}{233}) & (\textcolor{blue}{233}) & (188) & (176) \\ 
		\hline
	\end{tabular}
	\begin{itemize}
	    \item [] \footnotesize \textit{W-IID} (\textit{W-NID}): Wald-type methods with $\psi(u^q)$ estimated by $\hat{\psi}_I$ ($\hat{\psi}_N$); \textit{S-IID} (\textit{S-NID}): score methods with $\psi(u^q)$ estimated by $\hat{\psi}_I$ ($\hat{\psi}_N$); \textit{BOOT}: bootstrap method based on the joint estimation; \textit{COVES}: method in \cite{he2010detection}.
	\end{itemize}
\end{table}

\begin{figure}[H]
	\centering
	\includegraphics[width=15cm]{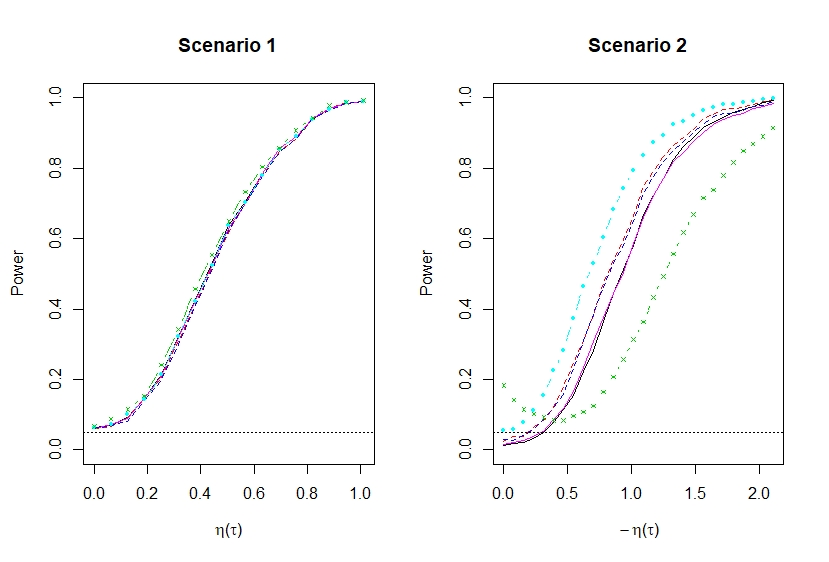}
	\includegraphics[width=15cm]{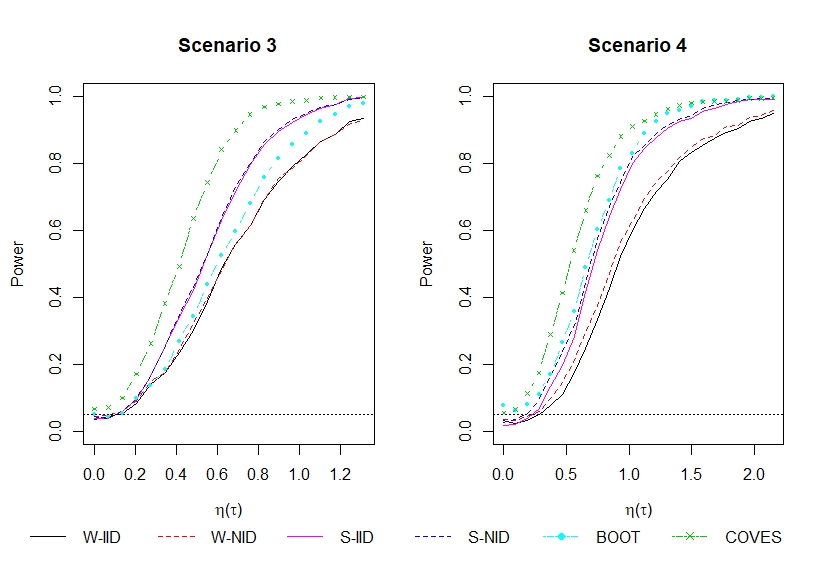}
	\caption{Power analysis plots for testing $H_0: \eta(\tau) = 0$ at $\tau = 0.8$ and $n = 100$.}
	\label{fig:PA_plots}
\end{figure}

\section{Real Data Analysis}\label{sec:real_data_ES}

\subsection{Application to 2018 CPS Income Data}\label{subsec:real_data_CPS_ES}
We illustrate the merit of the proposed score-type inference approach by analyzing a data set from the Current Population Survey (CPS) database, which can be assessed at \url{https://www.census.gov/programs-surveys/cps.html}. The CPS is a monthly survey of about 60,000 U.S. households conducted by the United States Census Bureau for the Bureau of Labor Statistics. Information collected in the survey includes employment status, income from work and a number of demographic characteristics. From the 2018 CPS Annual Social and Economic Supplement (ASEC) Bridge Files, we compile a data set, which contains 870 individuals (401 male and 469 female). To find out if there exists a pay gap between women and men, we use hourly wage in U.S. dollars as the response and consider \textit{Gender}, \textit{Age}, $\textit{Age}^2$, \textit{Education level (ordinal variable with three levels)} and \textit{Work status (full-times v.s. part-time)} as explanatory variables in the joint regression model.

To check whether the quantile error depends on the predictors, we conduct a heterogeneity analysis by analyzing the residual patterns. At a given quantile level $\tau$, we fit a linear quantile regression model and compare the variance of the quantile residuals for different covariate groups.
Table \ref{tab:res_var_edu} summarizes quantile residual variances associated with different education levels at $\tau = 0.7, 0.75$ and $0.8$. The result shows that the quantile residuals appear to depend on the \textit{Education level}. Therefore, for Wald and score inference approaches, we focus on \textit{W-NID} and \textit{S-NID} methods.
Moreover, the \textit{Education level} is an unbalanced covariate since the female group has a higher proportion of subjects with education level 3. 

\begin{table}[H]
\centering
\caption{The sample variance of quantile residuals for different education levels. Values inside the parentheses are standard errors obtained with jackknife.}
\label{tab:res_var_edu}
\begin{tabular}{cccc}
  \hline
 \multirow{2}{*}{Education level} & \multicolumn{3}{c}{$\tau$} \\ \cline{2-4}
  & 0.7 & 0.75 & 0.8 \\ 
  \hline
1 & 13.18 (2.17) & 14.03 (2.17) & 14.59 (2.09) \\ 
2 & 29.49 (3.03) & 29.49 (2.96) & 29.64 (2.96) \\ 
3 & 170.55 (39.25) & 170.49 (39.71) & 169.30 (39.52) \\ 
   \hline
\end{tabular}
\end{table}

Let $\theta_g^e$ denote the upper tail CES difference of the hourly wage between female and male groups. At the probability levels $0.7, 0.75$ and $0.8$, we apply the proposed inference method to test $H_0: \theta_g^e = 0$ against $H_a: \theta_g^e \ne 0$. Except for the \textit{COVES} method, all the other approaches suggest a significant pay gap between the two gender groups. We further calculate 95\% confidence intervals of $\theta_g^e$ using different methods and the results are summarized in Table \ref{tab:CI_re_CPS}.
Under the significance level of 5\%, results of \textit{W-NID}, \textit{S-NID} and the bootstrap approaches suggest that the female employee is substantially under-paid compared to the male employee, when other characteristics are kept the same. In contrast, the \textit{COVES} test may be negatively affected by the unbalanced covariate \textit{Education level} and thus fails to capture the tail difference.


\begin{table}[H]
	\centering
	\caption{The 95\% confidence intervals of $\theta_g^e$ given by different approaches, where $\theta_g^e$ is the CES difference of the hourly wage between female and male groups for the 2018 CPS income data.}
	\label{tab:CI_re_CPS}
	\begin{tabular}{cccccc}
		\hline
		$\tau$& \textit{W-NID} & \textit{S-NID} & \textit{BOOT} & \textit{COVES} \\ 
		\hline
		0.7 & (-3.09, -0.47) & (-3.27, -0.44) & (-3.24, -0.32) & \textbf{(-4.34, 0.12)} \\ 
		0.75 & (-3.31, -0.45) & (-3.41, -0.37) & (-3.39, -0.37) & \textbf{(-4.26, 0.6)} \\ 
		0.8 & (-3.35, -0.15) & (-3.63, -0.14) & (-3.45, -0.05) & \textbf{(-4.49, 1.14)} \\ 
		\hline
	\end{tabular}
\end{table}

\subsection{Power Analysis for the Opportunity Knocks Data}\label{subsec:real_data_OK_ES}
To further assess the finite sample performance of the proposed score test, we conduct a power analysis based on the ``Opportunity Knocks" (OK) experiment \citep{angrist2014opportunity}, which was designed to explore the effects of academic achievement awards for first-year and second-year college students. 
For our analysis, we consider a subset with all second-year students, which consists of 183 treated subjects and 337 untreated subjects. Treated students can receive bonus awards and have the opportunity to interact with randomly assigned peer advisors who can provide advice about study strategies, time management, and university bureaucracy.

The award scheme offered cash incentives to students with course grades above 70. Therefore, the academic performance of students can be measured by the amount they earned in the OK experiment.
To determine how the academic performance of students are motivated by the merit award, we define the response variable $Y$ as the earning of students (in 1000 U.S. dollars) from the OK program. Besides the treatment indicator $D$, we consider six additional covariates, including gender ($X_1$), high school grade ($X_2$), an indicator for English mother tongue ($X_3$), whether the student answers the scholarship formula question correctly ($X_4$, yes v.s. no), and mother's and father's education levels ($X_5$ and $X_6$, defined as above college degree or not). 
We apply the proposed inference method to test the treatment effect $\theta^e_D$ (upper tail ES regression coefficient associated with treatment variable $D$) at $\tau = 0.7, 0.75$ and $0.8$, and the results are summarized in Table \ref{tab:CI_OK}. The data indicates that there's a tendency that the merit award has a positive effect on the academic performance. However, none of the methods show that the treatment is statistically significant on the CES based on the original data, which might be due to the small sample size. To determine the sample size needed for different methods to capture the treatment difference, we conduct a power analysis.

In order to mimic the response distribution based on the linearity model assumption, we fit a linear quantile regression model using the original data at $\tau = 0.75$,
\begin{equation}
\label{equ:qr_OK}
    \hat{Q}_\tau(Y | D, X_1, \dots, X_6) = \hat{\beta}_0(\tau) + \hat{\beta}_1(\tau)D + \sum_{i=1}^6 \hat{\alpha}_i(\tau) X_i.
\end{equation}
We then obtain quantile residuals, defined as $\hat{\epsilon}_i(\tau) = Y_i - \hat{Q}_\tau(Y_i | D_i, X_1, \dots, X_6)$, and perform a heterogeneity analysis similar to the one in Section \ref{subsec:real_data_CPS_ES}. The results indicate that the residual term depends on $X_4$. Therefore, for power analysis, we focus on $S$-$NID$ and $W$-$NID$ approaches, and the response is generated by
\begin{equation}
    \label{equ:sim_res_OK}
    Y_{\text{sim}} = \hat{\beta}_0(\tau) + \hat{\beta}_1(\tau)D + \sum_{i=1}^6 \hat{\alpha}_i(\tau) \tilde{X}_i^d + \xi \cdot I(D = 1) + \tilde{\epsilon},
\end{equation}
where $\hat{\beta}_0(\tau), \hat{\beta}_1(\tau)$ and $\hat{\alpha}_i(\tau)$ are regression coefficient estimators in \eqref{equ:qr_OK}, $\tilde{X}_i^d$ follows the empirical distribution of covariate $X_i$ in group $D = d$, and $\tilde{\epsilon}$ is randomly sampled from the quantile residuals $\{\hat{\epsilon}_i(\tau)\}$ stratified by different grouped values of $X_4$. Since the original treatment effect is not statistically significant, we add an additional signal $\xi$ in \eqref{equ:sim_res_OK} to increase the treatment difference, and let the sample size of the two treatment groups be the same. 
We apply the proposed \textit{S-NID} method to the simulated data. Table \ref{tab:sample_size_OK} summarizes the sample size needed for different methods to reach a power of 0.9 at $\tau = 0.75$. The results show that the score test is clearly outperforming the Wald test and the bootstrap method, and the latter two require a trial with more subjects.

\begin{table}[H]
    \centering
    \caption{Point estimation and the 95\% confidence intervals (within the parentheses) of $\theta^e_D$ given by different approaches based on the original OK data, where $\theta^e_D$ is the ES regression coefficient associated with treatment variable $D$ and measures the treatment effect of the merit award on the upper tail CES.}
    \label{tab:CI_OK}
    \begin{tabular}{ccccc}
    \hline
    $\tau$ & \textit{S-NID} & \textit{W-NID} & \textit{BOOT} & \textit{COVES} \\ 
    \hline
    \multirow{2}{*}{0.70} & 0.298 & 0.375 & 0.375 & 0.215 \\
     & (-0.086, 0.651) & (-0.001, 0.751) & (-0.327, 1.078) & (-0.076, 0.506) \\
     \multirow{2}{*}{0.75} & 0.286 & 0.304 & 0.304 & 0.192 \\
     & (-0.126, 0.678) & (-0.116, 0.724) & (-0.379, 0.987) & (-0.124, 0.507) \\
     \multirow{2}{*}{0.80} & 0.278 & 0.204 & 0.204 & 0.207 \\
     & (-0.190, 0.726) & (-0.229, 0.637) & (-0.469, 0.878) & (-0.151, 0.564) \\
     \hline
    \end{tabular}
\end{table}

\begin{table}[H]
	\centering
	\caption{Sample size needed for each treatment group in the simulated OK data to reach power 0.9 at $\tau = 0.75$.}
	\label{tab:sample_size_OK}
	\begin{tabular}{ccccccc}
		\hline
		$\xi$ & \textit{S-NID} & \textit{W-NID} & \textit{BOOT} & \textit{COVES} \\ 
		\hline
		0.2 & 436 & 484 & 494 & 402 \\ 
		0.3 & 282 & 292 & 314 & 236 \\ 
		0.4 & 210 & 215 & 218 & 171 \\ 
		0.5 & 134 & 155 & 168 & 128 \\ 
		0.6 & 107 & 120 & 129 & 100 \\ 
		0.7 & 82 & 94 & 101 & 82 \\ 
		\hline
	\end{tabular}
\end{table}

\section{Conclusion}\label{sec:conclusion_ES}
In this paper, we considered the joint modeling of conditional quantile and ES. A two-step estimation procedure is proposed to reduce the computational effort.
We showed that the resulting two-step estimator is asymptotically equivalent to the joint estimator, but the former is numerically more efficient. In addition, the two-step estimator is locally robust to the perturbation of the quantile estimation in the first step.
We further developed a score-type inference method for hypothesis testing and confidence interval construction. The proposed score method is robust in performance, especially for cases with a large number of confounding factors and heterogeneous errors.

We chose parametric linear models for the joint-regression framework due to its computational efficiency and model interpretability. This framework can be further extended by considering more general models. \cite{wang2018semi}, \cite{taylor2019forecasting} and \cite{patton2019dynamic} consider dynamic models for ES with autoregressive features. To model CES based on exogenous covariates, another feasible alternative is to employ some nonparametric or semiparametric models, e.g., varying coefficient models \citep{hastie1993varying} and generalized additive models \citep{hastie1990generalized}. The proposed two-step estimation procedure and inference methods can be adapted accordingly, but further theoretical and practical investigations are needed. 

\section{Appendix}\label{sec:appendix_ES}
\subsection{Finite moment conditions} \label{subsec:moment_cond}
For some constant $c > 0$, define a neighborhood of $\boldsymbol{\theta}^q_0$ as $U_c(\boldsymbol{\theta}^q_0) = \{\boldsymbol{\theta}^q \in \boldsymbol{\Theta}^q: \|\boldsymbol{\theta}^q - \boldsymbol{\theta}^q_0\| \le c\}$. Similarly, we denote a neighborhood of $\boldsymbol{\theta}^e_0$ as $U_c(\boldsymbol{\theta}^e_0)$ and a neighborhood of $\boldsymbol{\theta}_0 = (\boldsymbol{\theta}^{q \prime}_0, \boldsymbol{\theta}^{e \prime}_0)^\prime$ as $U_c(\boldsymbol{\theta}_0)$. 

\begin{enumerate} \label{cond:moment}
    \item [($\mathcal{M}$-1)] 
    We assume the following moments are finite for a given constant $c > 0$: $E\big\{|G_1(Y)|\big\}$, $E\big\{|a(Y)|\big\}$, $E\big\{\|\mathbf{X}\|^r \sup_{\boldsymbol{\theta}^q \in U_c(\boldsymbol{\theta}^q_0)}|G_1^{(1)}(\mathbf{X}^\prime \boldsymbol{\theta}^q)|^r\big\}$ with $r = 1$ and 2, \\ $E\big\{\|\mathbf{X}\|^2 \sup_{\boldsymbol{\theta} \in U_c(\boldsymbol{\theta}_0)}|G_1^{(1)}(\mathbf{X}^\prime \boldsymbol{\theta}^q) G_2(\mathbf{X}^\prime \boldsymbol{\theta}^e)|\big\}$, \\
    $E\big\{\|\mathbf{X}\|^r \sup_{\boldsymbol{\theta}^e \in U_c(\boldsymbol{\theta}^e_0)}|G_2^{(1)}(\mathbf{X}^\prime \boldsymbol{\theta}^e)|^r\big\}$ with $r = 1$ and 2,\\ $E\big\{\|\mathbf{X}\|^{2r} \sup_{\boldsymbol{\theta}^e \in U_c(\boldsymbol{\theta}^e_0)}|G_2^{(1)}(\mathbf{X}^\prime \boldsymbol{\theta}^e)|^r\big\}$ with $r = 1$ and 2, \\ $E\big\{\|\mathbf{X}\| \sup_{\boldsymbol{\theta}^e \in U_c(\boldsymbol{\theta}^e_0)}|G_2^{(1)}(\mathbf{X}^\prime \boldsymbol{\theta}^e)|E(|Y| \big | \mathbf{X}) \big\}$, \\ $E\big[\|\mathbf{X}\|^3 \sup_{\boldsymbol{\theta}^e \in U_c(\boldsymbol{\theta}^e_0)}\big\{G_2^{(1)}(\mathbf{X}^\prime \boldsymbol{\theta}^e)\big\}^2E(|Y| \big | \mathbf{X}) \big]$ and \\ $E\big[\|\mathbf{X}\|^2 \sup_{\boldsymbol{\theta}^e \in U_c(\boldsymbol{\theta}^e_0)}\big\{G_2^{(1)}(\mathbf{X}^\prime \boldsymbol{\theta}^e)\big\}^2E(Y^2 \big | \mathbf{X}) \big]$.
\end{enumerate}

\subsection{Proofs of Theorems}\label{subsec:proof_of_thm}
Let $\omega_\tau(Y, \mathbf{X}, \boldsymbol{\theta}^q, \boldsymbol{\theta}^e)$ be the derivative of $\rho_\tau(Y, \mathbf{X}, \boldsymbol{\theta}^q, \boldsymbol{\theta}^e)$ w.r.t. $\boldsymbol{\theta}^e$. That is,
\begin{align*}
    \omega_\tau(Y, \mathbf{X}, \boldsymbol{\theta}^q, \boldsymbol{\theta}^e) &= \frac{\partial}{\partial (\boldsymbol{\theta}^e)^\prime} \rho_\tau(Y, \mathbf{X}, \boldsymbol{\theta}^q, \boldsymbol{\theta}^e) \\
    &= \mathbf{X} G_2^{(1)}(\mathbf{X}^\prime \boldsymbol{\theta}^e) \left\{\mathbf{X}^\prime \boldsymbol{\theta}^e - \mathbf{X}^\prime \boldsymbol{\theta}^q + \frac{(\mathbf{X}^\prime \boldsymbol{\theta}^q - Y) I(Y \le \mathbf{X}^\prime \boldsymbol{\theta}^q)}{\tau}\right\}.
\end{align*}
For simplicity, denote $\omega_\tau(Y_i, \mathbf{X}_i, \boldsymbol{\theta}^q, \boldsymbol{\theta}^e)$ by $\omega_i(\boldsymbol{\theta}^q, \boldsymbol{\theta}^e)$ for subject $i$. To derive the asymptotic behavior of $\sqrt{n} (\hat{\boldsymbol{\theta}}^e - \boldsymbol{\theta}^e_0)$, we first present and prove three lemmas.

\begin{lemma} \label{lemma:sep_asy_normality_lemma1}
 Assume the conditions in Theorem \ref{thm:sep_asy_normality} hold. Then we have \begin{equation}
 \label{equ:sep_asy_normality_lemma1}
     \frac{1}{n} \sum_{i=1}^n \frac{\partial}{\partial (\boldsymbol{\theta}^e)^\prime} \omega_i(\hat{\boldsymbol{\theta}}^q, \boldsymbol{\theta}^e) \big|_{\boldsymbol{\theta}^e = \boldsymbol{\theta}^e_0} \overset{P}{\to} E\big\{\mathbf{X} \mathbf{X}^\prime G_2^{(1)}(\mathbf{X}^\prime \boldsymbol{\theta}^e_0)\big\}.
 \end{equation}
\end{lemma}

\begin{proof}
By LLN (law of large numbers), we have \begin{equation*}
    \frac{1}{n} \sum_{i=1}^n \frac{\partial}{\partial (\boldsymbol{\theta}^e)^\prime} \omega_i(\hat{\boldsymbol{\theta}}^q, \boldsymbol{\theta}^e) \big|_{\boldsymbol{\theta}^e = \boldsymbol{\theta}^e_0} \overset{a.s.}{\to} E\left\{\frac{\partial}{\partial (\boldsymbol{\theta}^e)^\prime} \omega(Y, \mathbf{X}, \hat{\boldsymbol{\theta}}^q, \boldsymbol{\theta}^e) \big|_{\boldsymbol{\theta}^e = \boldsymbol{\theta}^e_0}\right\},
\end{equation*}
where \begin{align*}
    &\frac{\partial}{\partial (\boldsymbol{\theta}^e)^\prime} \omega(Y, \mathbf{X}, \hat{\boldsymbol{\theta}}^q, \boldsymbol{\theta}^e) \big|_{\boldsymbol{\theta}^e = \boldsymbol{\theta}^e_0} \\ 
    &= \mathbf{X} \mathbf{X}^\prime \left[G_2^{(1)}(\mathbf{X}^\prime \boldsymbol{\theta}^e_0) + G_2^{(2)}(\mathbf{X}^\prime \boldsymbol{\theta}^e_0)\left\{\mathbf{X}^\prime \boldsymbol{\theta}^e_0 - \mathbf{X}^\prime \hat{\boldsymbol{\theta}}^q + \frac{(\mathbf{X}^\prime \hat{\boldsymbol{\theta}}^q - Y) I(Y \le \mathbf{X}^\prime \hat{\boldsymbol{\theta}}^q)}{\tau}\right\}\right].
\end{align*}
It then sufficient to show that the expectation of the second term in the above equation is $o(1)$. Notice that
\begin{align*}
    & E\left[\mathbf{X} \mathbf{X}^\prime G_2^{(2)}(\mathbf{X}^\prime \boldsymbol{\theta}^e_0)\left\{\mathbf{X}^\prime \boldsymbol{\theta}^e_0 - \mathbf{X}^\prime \hat{\boldsymbol{\theta}}^q + \frac{(\mathbf{X}^\prime \hat{\boldsymbol{\theta}}^q - Y) I(Y \le \mathbf{X}^\prime \hat{\boldsymbol{\theta}}^q)}{\tau}\right\}\right] \\
    &= E\left[\mathbf{X} \mathbf{X}^\prime G_2^{(2)}(\mathbf{X}^\prime \boldsymbol{\theta}^e_0) \cdot E\left\{\mathbf{X}^\prime \boldsymbol{\theta}^e_0 - \mathbf{X}^\prime \hat{\boldsymbol{\theta}}^q + \frac{(\mathbf{X}^\prime \hat{\boldsymbol{\theta}}^q - Y) I(Y \le \mathbf{X}^\prime \hat{\boldsymbol{\theta}}^q)}{\tau} \big| \mathbf{X}\right\}\right],
\end{align*}
where \begin{align*}
    & E\big\{\mathbf{X}^\prime \boldsymbol{\theta}^e_0 - \mathbf{X}^\prime \hat{\boldsymbol{\theta}}^q + \frac{(\mathbf{X}^\prime \hat{\boldsymbol{\theta}}^q - Y) I(Y \le \mathbf{X}^\prime \hat{\boldsymbol{\theta}}^q)}{\tau} \big| \mathbf{X}\big\} \\
    &= E\left[\mathbf{X}^\prime (\boldsymbol{\theta}^q_0 - \hat{\boldsymbol{\theta}}^q) + \frac{1}{\tau} \big\{\mathbf{X}^\prime(\hat{\boldsymbol{\theta}}^q - \boldsymbol{\theta}^q_0)I(Y \le \mathbf{X}^\prime \hat{\boldsymbol{\theta}}^q) + (\mathbf{X}^\prime \boldsymbol{\theta}^q_0 - Y) (I(Y \le \mathbf{X}^\prime \hat{\boldsymbol{\theta}}^q) - I(Y \le \mathbf{X}^\prime \boldsymbol{\theta}^q_0)) \big\} \big| \mathbf{X} \right].
\end{align*}
For the above equation, it's easy to verify that the terms involving $\mathbf{X}^\prime (\boldsymbol{\theta}^q_0 - \hat{\boldsymbol{\theta}}^q)$ are $o_p(1)$ due to the consistency of $\hat{\boldsymbol{\theta}}^q$. Besides, by Taylor expansion we have
\begin{align*}
    E\big\{I(Y \le \mathbf{X}^\prime \hat{\boldsymbol{\theta}}^q) - I(Y \le \mathbf{X}^\prime \boldsymbol{\theta}^q_0) \big| \mathbf{X} \big\} &= F_{Y | \mathbf{X}} (\mathbf{X}^\prime \hat{\boldsymbol{\theta}}^q) - F_{Y | \mathbf{X}} (\mathbf{X}^\prime \boldsymbol{\theta}^q_0) \\
    &= f_{Y | \mathbf{X}} (\mathbf{X}^\prime \boldsymbol{\theta}^q_0) \mathbf{X}^\prime(\hat{\boldsymbol{\theta}}^q - \boldsymbol{\theta}^q_0) + o\big\{\mathbf{X}^\prime(\hat{\boldsymbol{\theta}}^q - \boldsymbol{\theta}^q_0)\big\}\\
    &= o_p(1).
\end{align*}
Therefore, with the assumption that the conditional distribution of $Y$ given $\mathbf{X}$ has finite second moment, the last term $E\big\{Y(I(Y \le \mathbf{X}^\prime \hat{\boldsymbol{\theta}}^q) - I(Y \le \mathbf{X}^\prime \boldsymbol{\theta}^q_0)) \big| \mathbf{X} \big\} = o_p(1)$ holds by Cauchy-Schwarz inequality.
\end{proof}

\begin{lemma} \label{lemma:sep_asy_normality_lemma2}
 Assume the conditions in Theorem \ref{thm:sep_asy_normality} hold. Then we have \begin{equation}
 \label{equ:sep_asy_normality_lemma2}
     \underset{\hat{\boldsymbol{\theta}}^q: ||\hat{\boldsymbol{\theta}}^q - \boldsymbol{\theta}_0^q|| \le c \cdot n^{-1/2}}{\sup}\left \| \frac{1}{\sqrt{n}} \sum_{i=1}^{n} \left[ \big\{\omega_i(\hat{\boldsymbol{\theta}}^q, \boldsymbol{\theta}_0^e) - \omega_i(\boldsymbol{\theta}_0^q, \boldsymbol{\theta}_0^e)\big\} - \big\{\lambda_i(\hat{\boldsymbol{\theta}}^q)  - \lambda_i(\boldsymbol{\theta}_0^q)\big\} \right] \right \| 
   \overset{P}{\to} 0,
 \end{equation}
 where $c$ is a positive constant and
   \begin{equation*}
       \lambda_i(\boldsymbol{\theta}^q) = E\big\{\omega(Y_i, \mathbf{X}_i, \boldsymbol{\theta}^q, \boldsymbol{\theta}_0^e) | \mathbf{X}_i\big\}.
   \end{equation*}
\end{lemma}

\begin{proof}
According to the Remark on page 410 of \cite{doukhan1995invariance}, we can obtain the stochastic equicontinuity of $n^{-1/2} \sum_{i=1}^n \big\{\omega_i(\boldsymbol{\theta}^q, \boldsymbol{\theta}^e_0) - \lambda_i(\boldsymbol{\theta}^q)\big\}$. That is, for any $\epsilon > 0$, there exist a $\delta_\epsilon > 0$ such that
\begin{align*}
    \underset{n \to \infty}{\limsup} P\Bigg(\underset{\hat{\boldsymbol{\theta}}^q: ||\hat{\boldsymbol{\theta}}^q - \boldsymbol{\theta}_0^q|| \le \delta_\epsilon}{\sup} \Bigg\|\frac{1}{\sqrt{n}} \sum_{i=1}^n \big[ \big\{\omega_i&(\hat{\boldsymbol{\theta}}^q, \boldsymbol{\theta}_0^e) - \lambda_i(\hat{\boldsymbol{\theta}}^q)\big\} - \big\{\omega_i(\boldsymbol{\theta}_0^q, \boldsymbol{\theta}_0^e) - \lambda_i(\boldsymbol{\theta}_0^q)\big\} \big] \Bigg\| > \epsilon \Bigg) < \epsilon,
\end{align*}
which implies the desired result.
\end{proof}

\begin{lemma} \label{lemma:sep_asy_normality_lemma3}
 Under the conditions in Theorem \ref{thm:sep_asy_normality}, we have \begin{equation}
 \label{equ:sep_asy_normality_lemma3}
     \frac{1}{\sqrt{n}}\sum_{i=1}^{n} \big\{\lambda_i(\hat{\boldsymbol{\theta}}^q) - \lambda_i(\boldsymbol{\theta}_0^q)\big\} = \frac{1}{\sqrt{n}}\sum_{i=1}^{n}\big[E\big\{\omega(Y_i, \mathbf{X}_i, \hat{\boldsymbol{\theta}}^q, \boldsymbol{\theta}_0^e) | \mathbf{X}_i\big\} - E\big\{\omega(Y_i, \mathbf{X}_i, \boldsymbol{\theta}^q_0, \boldsymbol{\theta}_0^e) | \mathbf{X}_i\big\}\big] \overset{P}{\to} 0.
 \end{equation}
\end{lemma}

\begin{proof}
Recall that \begin{align*}
    \lambda_i(\boldsymbol{\theta}^q) &= E\big\{\omega(Y_i, \mathbf{X}_i, \boldsymbol{\theta}^q, \boldsymbol{\theta}^e_0 )\big | \mathbf{X}_i\big\} \\
    &= \mathbf{X}_i G_2^{(1)}(\mathbf{X}^\prime_i \boldsymbol{\theta}^e_0) \left[\mathbf{X}^\prime_i \boldsymbol{\theta}^e_0 - \mathbf{X}^\prime_i \boldsymbol{\theta}^q + \tau^{-1}E\big\{ (\mathbf{X}^\prime\boldsymbol{\theta}^q - Y_i)I(Y_i \le \mathbf{X}^\prime\boldsymbol{\theta}^q)\big| \mathbf{X}_i\big\}\right] \\
    &= \mathbf{X}_i G_2^{(1)}(\mathbf{X}^\prime_i \boldsymbol{\theta}^e_0) \left[\mathbf{X}^\prime_i \boldsymbol{\theta}^e_0 - \mathbf{X}^\prime_i \boldsymbol{\theta}^q + \tau^{-1} \mathbf{X}^\prime\boldsymbol{\theta}^q F_{Y_i|\mathbf{X}_i}(\mathbf{X}^\prime\boldsymbol{\theta}^q) - \tau^{-1}E\big\{Y_iI(Y_i \le \mathbf{X}^\prime\boldsymbol{\theta}^q)\big| \mathbf{X}_i\big\}\right].
\end{align*}
Therefore we have
\begin{align}
\label{equ:lemma3_lambdadiff_ES}
    \lambda_i(\hat{\boldsymbol{\theta}}^q) - \lambda_i(\boldsymbol{\theta}^q_0) = \mathbf{X}_i G_2^{(1)}(\mathbf{X}^\prime_i \boldsymbol{\theta}^e_0) \left\{ \mathbf{X}_i^\prime(\boldsymbol{\theta}^q_0 - \hat{\boldsymbol{\theta}}^q) + \tau^{-1} k_i(\hat{\boldsymbol{\theta}}^q, \boldsymbol{\theta}_0) + \tau^{-1} l_i(\hat{\boldsymbol{\theta}}^q, \boldsymbol{\theta}_0)\right\}, 
\end{align}
where \begin{align*}
    k_i(\hat{\boldsymbol{\theta}}^q, \boldsymbol{\theta}_0) &= \mathbf{X}^\prime_i \hat{\boldsymbol{\theta}}^q F_{Y_i|\mathbf{X}_i}(\mathbf{X}^\prime\hat{\boldsymbol{\theta}}^q) - \mathbf{X}^\prime\boldsymbol{\theta}^q_0 F_{Y_i|\mathbf{X}_i}(\mathbf{X}^\prime\boldsymbol{\theta}^q_0),\\
    l_i(\hat{\boldsymbol{\theta}}^q, \boldsymbol{\theta}_0) &= E\big\{Y_iI(Y_i \le \mathbf{X}^\prime\boldsymbol{\theta}^q_0)\big| \mathbf{X}_i\big\} - E\big\{Y_iI(Y_i \le \mathbf{X}^\prime\hat{\boldsymbol{\theta}}^q)\big| \mathbf{X}_i\big\}.
\end{align*}
By Taylor expansion we have
\begin{align*}
    k_i(\hat{\boldsymbol{\theta}}^q, \boldsymbol{\theta}_0) &= \mathbf{X}^\prime_i \hat{\boldsymbol{\theta}}^q\big\{F_{Y_i|\mathbf{X}_i}(\mathbf{X}^\prime\hat{\boldsymbol{\theta}}^q) - F_{Y_i|\mathbf{X}_i}(\mathbf{X}^\prime\boldsymbol{\theta}^q_0)\big\} + F_{Y_i|\mathbf{X}_i}(\mathbf{X}^\prime\boldsymbol{\theta}^q_0) \mathbf{X}^\prime (\hat{\boldsymbol{\theta}}^q - \boldsymbol{\theta}^q_0)\\
    &= \mathbf{X}^\prime_i \hat{\boldsymbol{\theta}}^q f_{Y_i | \mathbf{X}_i}(\mathbf{X}^\prime\boldsymbol{\theta}^q_0) \mathbf{X}^\prime(\hat{\boldsymbol{\theta}}^q - \boldsymbol{\theta}^q_0) + \tau \mathbf{X}^\prime (\hat{\boldsymbol{\theta}}^q - \boldsymbol{\theta}^q_0) + O\big(\|\hat{\boldsymbol{\theta}}^q - \boldsymbol{\theta}^q_0\|^2\big),
\end{align*}
where the first term can be written as
\begin{align*}
& \mathbf{X}^\prime_i \hat{\boldsymbol{\theta}}^q f_{Y_i | \mathbf{X}_i}(\mathbf{X}^\prime\boldsymbol{\theta}^q_0) \mathbf{X}^\prime(\hat{\boldsymbol{\theta}}^q - \boldsymbol{\theta}^q_0)\\
    &= \mathbf{X}^\prime_i \boldsymbol{\theta}^q_0 f_{Y_i | \mathbf{X}_i}(\mathbf{X}^\prime\boldsymbol{\theta}^q_0) \mathbf{X}^\prime(\hat{\boldsymbol{\theta}}^q - \boldsymbol{\theta}^q_0) + \big\{\mathbf{X}^\prime(\hat{\boldsymbol{\theta}}^q - \boldsymbol{\theta}^q_0)\big\}^2 f_{Y_i | \mathbf{X}_i}(\mathbf{X}^\prime\boldsymbol{\theta}^q_0) \\
    &= \mathbf{X}^\prime_i \boldsymbol{\theta}^q_0 f_{Y_i | \mathbf{X}_i}(\mathbf{X}^\prime\boldsymbol{\theta}^q_0) \mathbf{X}^\prime(\hat{\boldsymbol{\theta}}^q - \boldsymbol{\theta}^q_0) + O\big(\|\hat{\boldsymbol{\theta}}^q - \boldsymbol{\theta}^q_0\|^2\big).
\end{align*}
Therefore,
\begin{align}
    \label{equ:lemma3_lambdadiff_term1_ES}
    k_i(\hat{\boldsymbol{\theta}}^q, \boldsymbol{\theta}_0) = \mathbf{X}^\prime_i \boldsymbol{\theta}^q_0 f_{Y_i | \mathbf{X}_i}(\mathbf{X}^\prime\boldsymbol{\theta}^q_0) \mathbf{X}^\prime(\hat{\boldsymbol{\theta}}^q - \boldsymbol{\theta}^q_0) + \tau \mathbf{X}^\prime (\hat{\boldsymbol{\theta}}^q - \boldsymbol{\theta}^q_0) + O\big(\|\hat{\boldsymbol{\theta}}^q - \boldsymbol{\theta}^q_0\|^2\big).
\end{align}
For the second term $l_i(\hat{\boldsymbol{\theta}}^q, \boldsymbol{\theta}_0)$, notice that $E\big(Y_iI(Y_i \le \mathbf{X}^\prime\boldsymbol{\theta}^q)\big| \mathbf{X}_i\big)$ is continuously differentiable for all $\boldsymbol{\theta}^q$ in some neighborhood $U(\boldsymbol{\theta}_0^q)$ around $\boldsymbol{\theta}_0^q$, since we assume $F_{Y_i|\mathbf{X}}$ has a density which is strictly positive, continuous and bounded in this area. So $\forall \boldsymbol{\theta}^q \in U(\boldsymbol{\theta}_0^q)$, we can choose $\boldsymbol{\theta}_1^q \in U(\boldsymbol{\theta}_0^q)$ such that $\mathbf{X}_i^\prime \boldsymbol{\theta}_1^q \le \mathbf{X}_i^\prime \boldsymbol{\theta}^q$, then
   	\begin{align}
   	\label{equ:quantile_partial_detivative_ES}
   	\frac{\partial}{\partial (\boldsymbol{\theta}^q)^\prime} &E\{Y_i I(Y_i \le \mathbf{X}_i^\prime \boldsymbol{\theta}^q) | \mathbf{X}_i\} \\
   	&= \frac{\partial}{\partial (\boldsymbol{\theta}^q)^\prime} E\{Y_i I(Y_i \le \mathbf{X}_i^\prime \boldsymbol{\theta}_1^q) | \mathbf{X}_i\} + \frac{\partial}{\partial (\boldsymbol{\theta}^q)^\prime} E\{Y_i I(\mathbf{X}_i^\prime \boldsymbol{\theta}_1^q < Y_i \le \mathbf{X}_i^\prime \boldsymbol{\theta}^q) | \mathbf{X}_i\} \nonumber\\ 
   	&= \frac{\partial}{\partial (\boldsymbol{\theta}^q)^\prime} \int_{-\infty}^{\mathbf{X}_i^\prime \boldsymbol{\theta}_1^q} u \ dF_{Y_i|\mathbf{X}_i}(u) + \frac{\partial}{\partial (\boldsymbol{\theta}^q)^\prime} \int_{\mathbf{X}_i^\prime \boldsymbol{\theta}_1^q}^{\mathbf{X}_i^\prime \boldsymbol{\theta}^q} u \ dF_{Y_i|\mathbf{X}_i}(u) \nonumber\\ 
   	&= \mathbf{X}_i^\prime(\mathbf{X}_i^\prime \boldsymbol{\theta}^q) f_{Y_i|\mathbf{X}_i}(\mathbf{X}_i^\prime \boldsymbol{\theta}^q).
   	\end{align}
It then follows that \begin{align}
    \label{equ:lemma3_lambdadiff_term2_ES}
    l_i(\hat{\boldsymbol{\theta}}^q, \boldsymbol{\theta}_0) = \mathbf{X}_i^\prime \boldsymbol{\theta}^q_0 f_{Y_i | \mathbf{X}_i}(\mathbf{X}^\prime\boldsymbol{\theta}^q_0) \mathbf{X}^\prime(\boldsymbol{\theta}^q_0 - \hat{\boldsymbol{\theta}}^q) + O\big(\|\hat{\boldsymbol{\theta}}^q - \boldsymbol{\theta}^q_0\|^2\big).
\end{align}
Substituting $k_i(\hat{\boldsymbol{\theta}}^q, \boldsymbol{\theta}_0)$ and $l_i(\hat{\boldsymbol{\theta}}^q, \boldsymbol{\theta}_0)$ in \eqref{equ:lemma3_lambdadiff_ES} by \eqref{equ:lemma3_lambdadiff_term1_ES} and \eqref{equ:lemma3_lambdadiff_term2_ES}, we have
\begin{equation*}
    \lambda_i(\hat{\boldsymbol{\theta}}^q) - \lambda_i(\boldsymbol{\theta}^q_0) = O\big(\|\hat{\boldsymbol{\theta}}^q - \boldsymbol{\theta}^q_0\|^2\big).
\end{equation*}
Together with the condition $\|\hat{\boldsymbol{\theta}}^q - \boldsymbol{\theta}^q_0\|^2 = O_p(n^{-1}) = o_p(n^{-1/2})$, we can obtain that \begin{align*}
    \frac{1}{\sqrt{n}} \sum_{i=1}^n \big\{\lambda_i(\hat{\boldsymbol{\theta}}^q) - \lambda_i(\boldsymbol{\theta}^q_0)\big\} = o_p(1).
\end{align*}
\end{proof}

\begin{proof}[Proof of Theorem \ref{thm:sep_asy_normality}]
Applying the Taylor expansion, we have \begin{align}
    \mathbf{0} &= \frac{1}{n} \sum_{i=1}^n \omega_i(\hat{\boldsymbol{\theta}}^q, \hat{\boldsymbol{\theta}}^e)  \nonumber \\
    &= \frac{1}{n} \sum_{i=1}^n \omega_i(\hat{\boldsymbol{\theta}}^q, \boldsymbol{\theta}^e_0) + \frac{1}{n} \sum_{i=1}^n \frac{\partial}{\partial (\boldsymbol{\theta}^e)^\prime} \omega_i(\hat{\boldsymbol{\theta}}^q, \boldsymbol{\theta}^e) \big|_{\boldsymbol{\theta}^e = \boldsymbol{\theta}^e_0} \cdot (\boldsymbol{\theta}^e_0 - \hat{\boldsymbol{\theta}}^e) + R_n, \nonumber
\end{align}
where $\boldsymbol{\theta}^e_0$ is the true parameter vector and $R_n$ is a reminder term such that $\sqrt{n} R_n \to \mathbf{0}$ as $n \to \infty$. It then follows that \begin{align*}
    \sqrt{n} (\hat{\boldsymbol{\theta}}^e - \boldsymbol{\theta}^e_0) = &\left[\frac{1}{n} \sum_{i=1}^n \frac{\partial}{\partial (\boldsymbol{\theta}^e)^\prime} \omega_i(\hat{\boldsymbol{\theta}}^q, \boldsymbol{\theta}^e) \big|_{\boldsymbol{\theta}^e = \boldsymbol{\theta}^e_0}\right]^{-1} \\
    &\times \left[\frac{1}{\sqrt{n}} \sum_{i=1}^n \omega_i(\boldsymbol{\theta}^q_0, \boldsymbol{\theta}^e_0) + \frac{1}{\sqrt{n}} \sum_{i=1}^n \big\{\omega_i(\hat{\boldsymbol{\theta}}^q, \boldsymbol{\theta}^e_0) - \omega_i(\boldsymbol{\theta}^q_0, \boldsymbol{\theta}^e_0) \big\} + \sqrt{n} R_n\right].
\end{align*}

Combining this with Lemmas 1-3, the result of Theorem \ref{thm:sep_asy_normality} follows by
\begin{equation*}
    \sqrt{n}(\hat{\boldsymbol{\theta}}^e - \boldsymbol{\theta}^e_0) \overset{d}{\simeq} \left[E\big\{\mathbf{X} \mathbf{X}^\prime G_2^{(1)}(\mathbf{X}^\prime \boldsymbol{\theta}^e_0)\big\}\right]^{-1}  \left\{\frac{1}{\sqrt{n}} \sum_{i=1}^n \omega_i(\boldsymbol{\theta}^q_0, \boldsymbol{\theta}^e_0)\right\} = N(\mathbf{0}, \Lambda^{-1} \Omega \Lambda^{-1}).
\end{equation*}
\end{proof}

\begin{proof}[Proof of Theorem \ref{thm:joint_var_consistency}]
Based on the Taylor expansion on $\hat{\Lambda}$, and application of the asymptotic normality of $\hat{\boldsymbol{\theta}}^e$ and Slutsky's theorem, we have
 \begin{align*}
     \hat{\Lambda} &= n^{-1} \sum_{i=1}^n (\mathbf{X}_i \mathbf{X}_i^\prime) \cdot \left[G_2^{(1)} (\mathbf{X}_i^\prime \boldsymbol{\theta}^e_0) + \big\{G_2^{(1)} (\mathbf{X}_i^\prime \hat{\boldsymbol{\theta}}^e) - G_2^{(1)} (\mathbf{X}_i^\prime \boldsymbol{\theta}^e_0)\big\}\right] \\
     &= n^{-1} \sum_{i=1}^n (\mathbf{X}_i \mathbf{X}_i^\prime) \cdot G_2^{(1)} (\mathbf{X}_i^\prime \boldsymbol{\theta}^e_0) + n^{-1} \sum_{i=1}^n (\mathbf{X}_i \mathbf{X}_i^\prime) \left\{G_2^{(2)} (\mathbf{X}_i^\prime \boldsymbol{\theta}^e_0) \mathbf{X}_i^\prime (\hat{\boldsymbol{\theta}}^e - \boldsymbol{\theta}^e_0) + O\big(\|\hat{\boldsymbol{\theta}}^e - \boldsymbol{\theta}^e_0\|^2\big)\right\}\\
     & \overset{P}{\to} \Lambda.
 \end{align*}
Similar arguments yields the result for $\hat{\Omega}$.
\end{proof}

The proof of Theorem \ref{thm:score_stat_null} requires the following Lemmas 4-5.
 
 \begin{lemma}\label{lemma:score_stat_null_lemma1}
  Let $\boldsymbol{\theta}^q_0$ and $\boldsymbol{\theta}^e_{10}$ be the true parameters under $H_0$, with the assumptions in Theorem \ref{thm:score_stat_null},
  we have \begin{align*}
      \underset{\|
   		(\boldsymbol{\theta}^{q \prime}, \boldsymbol{\theta}_1^{e \prime})^\prime
   		- (\boldsymbol{\theta}^{q \prime}_0, \boldsymbol{\theta}_{10}^{e \prime})^\prime \| \le c \cdot n^{-1/2}}{\sup} \left\| S_n(\boldsymbol{\theta}^q, \boldsymbol{\theta}_1^e) - S_n(\boldsymbol{\theta}^q_0, \boldsymbol{\theta}_{10}^e) - E\big\{S_n(\boldsymbol{\theta}^q, \boldsymbol{\theta}_{1}^e) - S_n(\boldsymbol{\theta}^q_0, \boldsymbol{\theta}_{10}^e) \big| \mathbf{X}_i\big\} \right\| = o_p(1).
  \end{align*}
 \end{lemma}
 
 \begin{proof}
Applying similar arguments in the proof of Lemma \ref{lemma:sep_asy_normality_lemma2}, the result follows by the stochastic equicontinuity of $S_n(\boldsymbol{\theta}^q, \boldsymbol{\theta}_1^e) - E\big\{S_n(\boldsymbol{\theta}^q, \boldsymbol{\theta}_1^e) \big| \mathbf{X}_i\big\}$.
\end{proof}
 
 \begin{lemma}\label{lemma:score_stat_null_lemma2}
  With the assumptions in Theorem \ref{thm:score_stat_null}, under $H_0$, we have 
  \begin{align*}
      E\big\{S_n(\hat{\boldsymbol{\theta}}^q, \hat{\boldsymbol{\theta}}_1^e) - S_n(\boldsymbol{\theta}^q_0, \boldsymbol{\theta}_{10}^e) \big| \mathbf{X}_i\big\} = o_p(1).
  \end{align*}
 \end{lemma}
 
 \begin{proof}
 The difference can be written as \begin{align*}
     &E\big\{S_n(\hat{\boldsymbol{\theta}}^q, \hat{\boldsymbol{\theta}}_1^e) - S_n(\boldsymbol{\theta}^q_0, \boldsymbol{\theta}_{10}^e) \big| \mathbf{X}_i\big\} \\
     &= \frac{1}{\sqrt{n}} \sum_{i=1}^n \mathbf{Z}_i^* G_2^{(1)}(\mathbf{X}_i^\prime \boldsymbol{\theta}_{0}^e) \left\{\mathbf{W}_i^\prime (\hat{\boldsymbol{\theta}}^e_1 - \boldsymbol{\theta}^e_{10}) - \mathbf{X}_i^\prime(\hat{\boldsymbol{\theta}}^q - \boldsymbol{\theta}^q_0) + \tau^{-1}k_i(\hat{\boldsymbol{\theta}}^q, \boldsymbol{\theta}_0) + \tau^{-1} l_i(\hat{\boldsymbol{\theta}}^q, \boldsymbol{\theta}_0)\right\},
 \end{align*}
 where the functions $k_i$ and $l_i$ are defined in \eqref{equ:lemma3_lambdadiff_ES}. Based on \eqref{equ:lemma3_lambdadiff_term1_ES} and \eqref{equ:lemma3_lambdadiff_term2_ES}, the difference can be simplified to \begin{align*}
     &E\big\{S_n(\hat{\boldsymbol{\theta}}^q, \hat{\boldsymbol{\theta}}_1^e) - S_n(\boldsymbol{\theta}^q_0, \boldsymbol{\theta}_{10}^e) \big| \mathbf{X}_i\big\} \\
     &= \frac{1}{\sqrt{n}} \sum_{i=1}^n \mathbf{Z}_i^* G_2^{(1)}(\mathbf{X}_i^\prime \boldsymbol{\theta}_{0}^e) \mathbf{W}_i^\prime (\hat{\boldsymbol{\theta}}^e_1 - \boldsymbol{\theta}^e_{10}) + o_p(1)\\
     &= \left\{\frac{1}{n} \sum_{i=1}^n \mathbf{Z}_i^* G_2^{(1)}(\mathbf{X}_i^\prime \boldsymbol{\theta}_{0}^e) \mathbf{W}_i^\prime\right\} \sqrt{n}(\hat{\boldsymbol{\theta}}^e_1 - \boldsymbol{\theta}^e_{10}) + o_p(1) \\
     &= \left[\Pi_Z^\prime G \Pi_W - \Pi_Z^\prime G \Pi_W (\Pi_W^\prime G \Pi_W)^{-1} \Pi_W^\prime G \Pi_W \right] \sqrt{n}(\hat{\boldsymbol{\theta}}^e_1 - \boldsymbol{\theta}^e_{10}) + o_p(1) = o_p(1).
 \end{align*}
 The last equality holds due to the orthogonal transformation given in \eqref{equ:ortho_score}.
\end{proof}
 
\begin{proof}[Proof of Theorem \ref{thm:score_stat_null}]
 For any $(\boldsymbol{\theta}^{q \prime}, \boldsymbol{\theta}^{e \prime}_1)^\prime$ such that $\|(\boldsymbol{\theta}^{q \prime}, \boldsymbol{\theta}^{e \prime}_1)^\prime - (\boldsymbol{\theta}^{q \prime}_0, \boldsymbol{\theta}^{e \prime}_{10})^\prime\| \le c \cdot n^{-1/2}$, define \begin{align*}
     S_n(\boldsymbol{\theta}^q, \boldsymbol{\theta}^e_1) = \frac{1}{\sqrt{n}} \sum_{i=1}^n \mathbf{Z}^*_i G^{(1)}(\mathbf{X}^\prime_i \boldsymbol{\theta}^e_{0}) \left\{\mathbf{W}^\prime_i \boldsymbol{\theta}^e_1 - \mathbf{X}^\prime_i \boldsymbol{\theta}^q + \tau^{-1} (\mathbf{X}^\prime_i \boldsymbol{\theta}^q - Y_i)I_(Y_i \le \mathbf{X}^\prime_i \boldsymbol{\theta}^q)\right\}.
 \end{align*}
 Under $H_0: \boldsymbol{\theta}^e_2 = \mathbf{0}$, $\|\mathbf{Z}^*_i G^{(1)}(\mathbf{X}^\prime_i \boldsymbol{\theta}^e_{0}) - \hat{\mathbf{Z}}^*_i G^{(1)}(\mathbf{X}^\prime_i \hat{\boldsymbol{\theta}}^e)\| = O_p(\|\hat{\boldsymbol{\theta}}^e - \boldsymbol{\theta}^e_{0}\|) = O_p(n^{-1/2})$. Therefore, if we have the asymptotic normality of $S_n(\hat{\boldsymbol{\theta}}^q, \hat{\boldsymbol{\theta}}^e_1)$, then $$
 S_n - S_n(\hat{\boldsymbol{\theta}}^q, \hat{\boldsymbol{\theta}}^e_1) = O_p(n^{-1/2}) O_p(1) = o_p(1),
 $$ which implies $S_n \overset{d}{\simeq} S_n(\hat{\boldsymbol{\theta}}^q, \hat{\boldsymbol{\theta}}^e_1)$.
 
 According to Lemmas 4-5, we can obtain that $S_n \overset{d}{\simeq} S_n(\boldsymbol{\theta}_0^q, \boldsymbol{\theta}_{10}^e)$, where \begin{align*}
    S_n(\boldsymbol{\theta}_0^q, \boldsymbol{\theta}_{10}^e) &= \frac{1}{\sqrt{n}} \sum_{i=1}^{n} \mathbf{Z}_i^* G_2^{(1)}(\mathbf{X}_i^\prime \boldsymbol{\theta}_{0}^e) \big\{\mathbf{W}_i^\prime\boldsymbol{\theta}_{10}^e - \mathbf{X}_i^\prime \boldsymbol{\theta}_0^q + \tau^{-1} (\mathbf{X}_i^\prime \boldsymbol{\theta}_0^q - Y_i) I(Y_i \le \mathbf{X}_i^\prime \boldsymbol{\theta}_0^q)\big\} \\
   &= AN(0, \Sigma),
   \end{align*}
   and
   \begin{align*}
   \Sigma &= \frac{1}{n} \sum_{i=1}^n \left[\mathbf{Z}_i^* \mathbf{Z}_i^{*\prime} \big\{G_2^{(1)}(\mathbf{W}_i^\prime \boldsymbol{\theta}_{10}^e)\big\}^2 \frac{1}{\tau^2} \text{Var}\big\{\epsilon_i^q I(\epsilon_i^q \le 0) | \mathbf{X}_{i} \big\}\right],
   \end{align*}
   here $\epsilon_i^q = Y_i - \mathbf{X}_i^\prime \boldsymbol{\theta}_0^q$ are the quantile residuals. Similar to the arguments in the proof of Theorem \ref{thm:joint_var_consistency}, we have $\hat{\Sigma}_n(\hat{\psi}) \overset{P}{\to} \Sigma$, and therefore the desired result under $H_0$ follows.
   
   Under the local alternative $H_n: \boldsymbol{\theta}_2^e = \boldsymbol{\theta}^e_{20} / \sqrt{n}$, notice that
 \begin{align*}
     &E\big\{\mathbf{W}_i^\prime\boldsymbol{\theta}_{10}^e - \mathbf{X}_i^\prime \boldsymbol{\theta}_0^q + \tau^{-1} (\mathbf{X}_i^\prime \boldsymbol{\theta}_0^q - Y_i) I(Y_i \le \mathbf{X}_i^\prime \boldsymbol{\theta}_0^q) \big | \mathbf{X}_i\big\} \\
     &= \mathbf{W}_i^\prime\boldsymbol{\theta}_{10}^e - \mathbf{X}_i^\prime \boldsymbol{\theta}_0^q + \mathbf{X}_i^\prime \boldsymbol{\theta}_0^q F_{Y_i | \mathbf{X}_i}(\mathbf{X}_i^\prime \boldsymbol{\theta}_0^q)- \tau^{-1} E\big\{Y_i I(Y_i \le \mathbf{X}_i^\prime \boldsymbol{\theta}_0^q)\big\} \\
     &= -\mathbf{Z}_i^\prime \boldsymbol{\theta}^e_{20} / \sqrt{n}.
 \end{align*}
 The result follows since $S_n \overset{d}{\simeq} S_n(\boldsymbol{\theta}_0^q, \boldsymbol{\theta}_{10}^e)$ under $H_n$.
\end{proof}

\bibliography{ref.bib}
\end{document}